\documentclass[12pt]{iopartF}

 \expandafter\let\csname equation*\endcsname\relax
 \expandafter\let\csname endequation*\endcsname\relax

\usepackage{amscd,amssymb,amsmath,amsthm}

\usepackage{epsfig}
\usepackage{graphicx,xspace}
\usepackage{pstricks}
\usepackage[left=3cm,top=3cm,right=3cm,bottom=3cm]{geometry}
\usepackage{color}
\definecolor{labelkey}{cmyk}{.4,.2,0,0}

\newcommand{\ii}{\mathbf i}
\newcommand{\eps}{\varepsilon}
\renewcommand{\r}{\mathbf{r}}
\newtheorem{theorem}{Theorem}
\newtheorem{proposition}[theorem]{Proposition}
\newtheorem{lemma}[theorem]{Lemma}
\theoremstyle{remark}
\newtheorem{remark}[theorem]{Remark}

\renewcommand{\Pr}{\mathbb P}
\renewcommand*\contentsname{}

\begin{document}

\newcommand \be  {\begin{equation}}
\newcommand \ee  {\end{equation}}
\newcommand \bea {\begin{eqnarray} }
\newcommand \eea {\end{eqnarray}}

\title[On the position of the maximum of log-correlated processes]{Moments of the position of the maximum for GUE characteristic polynomials and for log-correlated Gaussian processes.\footnote[1]{\large with Appendix A written by Alexei Borodin and Vadim Gorin. }}

\vskip 0.2cm

\author{Yan V. Fyodorov}

\address{School of Mathematical Sciences, Queen Mary University of London\\  London E1 4NS, United Kingdom}

\author{Pierre Le Doussal 
}

\address{CNRS-Laboratoire de Physique Th\'eorique de l'Ecole Normale Sup\'erieure\\
24 rue Lhomond, 75231 Paris
Cedex-France}

\date{Received:  / Accepted:  / Published }

\begin{abstract}
We study three instances of log-correlated processes on the interval:   the logarithm of the Gaussian unitary ensemble (GUE) characteristic polynomial, the Gaussian log-correlated potential in presence
of edge charges, and the Fractional Brownian motion with Hurst index $H \to 0$ (fBM0).
In previous
collaborations we obtained the probability distribution function (PDF) of the value of the
global minimum (equivalently maximum) for the first two processes, using the {\it freezing-duality conjecture} (FDC).
Here we study the PDF of the position of the maximum
$x_m$ through its moments. Using replica, this requires  calculating moments of
the density of eigenvalues in the $\beta$-Jacobi ensemble.
Using Jack polynomials  we obtain an exact and explicit expression for both positive and negative integer moments
 for arbitrary $\beta >0$ and positive integer $n$ in terms of sums over partitions.
 For positive moments, this expression agrees with a very recent independent derivation by Mezzadri and Reynolds.
We check our results against a contour integral formula
derived recently by Borodin and Gorin (presented in the Appendix A from these authors).
The duality necessary for the FDC to work is proved,
and on our expressions, found to correspond to exchange of partitions with their dual.
Performing the limit $n \to 0$ and to negative Dyson index $\beta \to -2$,
we obtain the moments of $x_m$
and give explicit expressions for
the lowest ones. Numerical checks for the GUE polynomials,
performed independently by N. Simm,
indicate encouraging agreement. Some results are also obtained
for moments in Laguerre, Hermite-Gaussian, as well as
circular and related ensembles. The correlations of the position and the value of
the field at the minimum are also analyzed.

\end{abstract}

\maketitle

\tableofcontents

\section{Introduction}

Logarithmically correlated  Gaussian (LCG) random processes and fields attract growing attention in Mathematical Physics and Probability and play an important role in problems of Statistical Mechanics, Quantum Gravity,  Turbulence, Financial Mathematics and Random Matrix Theory,  see e.g. recent papers \cite{FKS13}, \cite{FK14} and \cite{DRSV14,RV2014,NKistler} for introduction and some background references
and \cite{CLD01} for earlier review including Condensed Matter applications.
A general lattice version of logarithmically correlated Gaussian field is
   a collection of Gaussian variables $V_{N,x} : x \in D_N$ attached to the sites of $d-$dimensional box $D_N$ of side length $N$ (assuming lattice spacing one) and characterized by the mean zero
    and the covariance structure
\begin{eqnarray}\label{LCdefa}
&&\mathbb{E}\{V^2_{N,x}\}=2g^2\log{N}+ f(x), \\ \label{LCdefb}
&&\mathbb{E}\{V_{N,x},V_{N,y}\}= 2g^2\log_{+}\frac{N}{|x - y|} + \psi(x,y), \quad \mbox{ for}\,\,  \,\, x \ne y \in D_N
\end{eqnarray}
where $\ln_+(w) = \max\left(\ln{w}, 0\right)$, $g>0$ and both $f(x)$ and $\psi(x,y)$  are bounded function far enough from the boundary of $D_N$. One also can define the continuous versions $V(x)$ of LCG fields on various domains $D\in \mathbb{R}^d$ which is then necessarily a random generalized function ("random distribution"), the most famous example being the Gaussian Free field in $d=2$, see \cite{Sheffield} for a rigorous definition. For $d=1$ the one-dimensional versions of LGC processes are known under the name of $1/f$ noises, see e.g. \cite{1fpower,FLDR12}. They appear frequently in Physics and Engineering sciences, and
also are rich and important mathematical objects of interest on their own. Such processes emerge, for example, in constructions of conformally invariant planar random curves \cite{AJKS11} and are relevant in random matrix theory and studies of the Riemann zeta-function on the critical line \cite{FK14}.

In particular, the problem of characterizing the distribution of the global maximum $M_N =\max_{x\in D_N}V_{N,x}$ of LCG fields and processes (or their continuum analogues) recently attracted a lot of
interest, in physics, see \cite{CLD01,FB08,FLDR09,FHK12,CRS,FG15}
and mathematics,
see \cite{BZ,BDZ,BK12,FK14,DZ14,DRZ2015,FyoSim15,AZ14,vargasfreezing,AZ15,AO15,ABH15}.
The distribution is proved to be given by the  Gumbel distribution {\it with random shift} \cite{DRZ2015}
and has a universal tail predicted by renormalization group arguments in \cite{CLD01}. The detail of the full distribution are not universal
and depend on some details of the behaviour of the covariance (\ref{LCdefb})
 for global $|x - y|\sim N$ scale 
 as well as on the subleading term $f(x)$ in the variance (\ref{LCdefa}).
 The explicit forms for the maximum distribution were conjectured
 in a few specific models of $1/f$ noises \cite{FB08,FLDR09,FK14,FyoSim15}.

 The goal of this paper is to provide some information about the distribution of the {\it position} of the global
 {\it minimum}
 \begin{equation}\label{argmax}
 x_{m}=\left\{x\in D \, : \, V(x)=\min_{y\in D}V(y)\right\}:= \mbox{Arg min}_{x\in D}V(x)
\end{equation}
  for some examples of 1-dimensional processes with logarithmic correlations, though depending on applications, one can be interested instead in a {\it maximum}.  Statistical properties of the value and position for maxima and minima
  are obviously trivially related in cases when $V(x) = -V(x)$ {\it in law}. \\

   Our first example is the modulus of the characteristic polynomial of a random GUE matrix over the interval $[-1,1]$ of the spectral parameter.
    As is well-known, in the limit of large sizes of the matrix the logarithm of that modulus is very intimately related to $1/f$ noises \cite{FKS13,FK14,FyoSim15,Webb14,Webb2015}.  In that example the interesting quantity is obviously the statistics of the maximum, with minimum value being trivially zero at every characteristic root of the matrix.  The second example is a general two-parameter variant of a log-correlated process on the interval  with, in the language of Coulomb gases, endpoint charges, introduced and studied in \cite{FLDR09}. That case may include a non-random (logarithmic) {\it background potential} $V_0(x)$, so that for the sum  $V_0(x) + V(x)$ we have
  \[
  \mbox{min}[ V_0(x) + V(x) ] = - \mbox{max}[ - V_0(x) + V(x) ], \quad  \mbox{in law}
  \]
     Our last example is a regularized version of the fractional Brownian motion with zero Hurst index, which is a {\it bona fide} (nonstationary) $1d$ LCG process \cite{FKS13}. Here statistics of maxima and minima are trivially related by symmetry in law.

\section{Models, method and main results}

We now define the three models to be considered. Although we did not yet succeed in obtaining
the full probability distribution function (PDF), ${\cal P}(x_m)$, of the position of the global minimum $x_m$, we derive formulae for all positive integer moments
of $x_m$ in terms of sums over partitions. In this section we present explicit values for
some low moments of $x_m$, more results can be found in the remainder of the paper.

As is clear from \cite{FLDR09} there is an intimate relation between statistics of extrema in log-correlated fields on an interval and the $\beta$-Jacobi ensemble of random matrices
\cite{DE2002,Forbook}. In the course of
our calculation we present methods to calculate the moments of the eigenvalue density of
the Jacobi ensemble. In this section we provide a very explicit { formula} for these moments
derived in remainder of the paper.
In particular our result is in agreement with recent results by Mezzadri and Reynolds
\cite{MR15}. We check our results against a contour integral { representation}
derived recently by Borodin and Gorin (presented in the Appendix A from these authors)
which, remarkably, also allows to calculate negative moments.

Finally we sketch the replica method and the application of the freezing-duality conjecture
to extract the moments of $x_m$ from the moments of the Jacobi ensemble.

\subsection{Results for log-correlated processes}

\subsubsection{GUE characteristic polynomials { (GUE-CP)}:}

\noindent

\medskip

Our first prediction is for
the lowest moments of the position of the global {\it maximum} for the the modulus of the characteristic polynomial $p_{N}(x) = \mathrm{det}(xI-H)$ of the Hermitian $N\times N$ matrix $H$ sampled with the probability weight
\begin{equation}
\label{guedensity}
P(H) \propto \mathrm{exp}(-2N\mathrm{Tr}(H^{2}))
\end{equation}
 known as the \textit{Gaussian Unitary Ensemble} (or GUE)\cite{AGZ09,Meh04,PS11}.  Here the variance is chosen to ensure that asymptotically for $N\to \infty$, the limiting mean density of the GUE eigenvalues is given by the Wigner semicircle law $\rho(x) = (2/\pi)\sqrt{1-x^{2}}$ supported in the interval $x\in [-1,1]$. Hence the object we want to study is $\log|p_{N}(x)| = \sum_i \ln|x - \lambda_i|$ where
the $\lambda_i$ are the eigenvalues of the GUE matrix $H$. To study its fluctuations it turns out to be
more convenient to subtract its mean. This leads to the following\\

\noindent{\bf Prediction 1.}
{\it Define $\phi_{N}(x)=2\log|p_{N}(x)|-2\mathbb{E}(\log|p_{N}(x)|)$ and consider the random variable
\begin{equation}
x_{m}^{(N)} := \mbox{Arg max}_{x \in [-1,1]} \phi_{N}(x)
 \label{maxrv}
\end{equation}
Then the lowest even integer moments of this random variable have the values
\begin{eqnarray}
 \lim_{N\to \infty}\mathbb{E}\left\{\left[x_m^{(N)}\right]^2\right\}= \frac{13}{49},   \quad \lim_{N \to \infty}\mathbb{E}\left\{\left[x_m^{(N)}\right]^4\right\}=\frac{20}{147}
\end{eqnarray}
whereas the odd integer moments  vanish by symmetry.

In particular, the kurtosis of the distribution of $x_{m}^{(N)}$ in the large-$N$ limit is given by
\begin{equation}
\lim_{N\to \infty}\left(\frac{\mathbb{E}\left\{\left[x_m^{(N)}\right]^4\right\}}
{\mathbb{E}\left\{\left[x_m^{(N)}\right]^2\right\}^2}-3\right)=-\frac{541}{507}\approx - 1.067\ldots
\end{equation}
}

To make a contact between $|p_N(x)|$ and the LCG processes we refer to the paper \cite{FKS13}.
That work revealed  that the natural large-$N$ limit of
$\phi_N(x)$  is given by the random Chebyshev-Fourier series
\begin{equation}\label{1/fch}
F(x) = -2\sum_{n=1}^{\infty}\frac{1}{\sqrt{n}}\, a_{n}\, T_{n}(x), \qquad x \in (-1,1),
\end{equation}
 with $T_{n}(x) = \cos(n\arccos(x))$ being Chebyshev polynomials and real $a_n$ being independent standard Gaussian variables. A quick computation shows that the covariance structure associated with the generalized process $F(x)$ is given by an integral operator with kernel
\begin{equation}\label{covop}
 \mathbb{E}\{F(x)F(y)\} = 4\sum_{n=1}^{\infty}\frac{1}{n}T_{n}(x)T_{n}(y) = -2\log(2|x-y|),
\end{equation}
 as long as $x\ne y$. Such a limiting process $F(x)$ is an example of an aperiodic $1/f$-noise.

  Note however that the series (\ref{1/fch}) is formal and diverges with probability one. In fact it should be understood as a random generalized function (distribution).  Though there is no sense in discussing the maxima and its position for generalized functions, the problem  is well-defined for the logmod of the characteristic polynomial $\log{|p_N(x)|}$ for any finite $N$. One therefore needs to find a tool to utilize its asymptotically Gaussian nature evident in (\ref{1/fch}). It turns out that the latter is encapsulated in the following asymptotic formula due to Krasovsky  \cite{K07} \footnote{See also earlier works \cite{ForFra,Gar05} where such formula was anticipated and proved for positive integer $\delta_j$} which will be central for our considerations:
\begin{eqnarray} \label{krasov}
\fl \mathbb{E}\left(\prod_{j=1}^{k}|p_{N}(x_{j})|^{2\delta_{j}}\right) &= \prod_{j=1}^{k}C(\delta_{j})(1-x_{j}^{2})^{\delta_{j}^{2}/2}(N/2)^{\delta_{j}^{2}}
e^{(2x_{j}^{2}-1-2\log(2))\delta_{j}N}\\
&\times \exp \Big[ -\sum_{1 \leq i < j \leq k}2\delta_{i}\delta_{j}\,\log|2(x_{i}-x_{j})|\Big] \left[1+O\left(\frac{\log
N}{N}\right)\right]\nonumber
\end{eqnarray}
  where $C(\delta) := 2^{2\delta^{2}}\frac{G(\delta+1)^{2}}{G(2\delta+1)}$, with
 $G(z)$ being the Barnes G-function. In particular, differentiating with respect to $\delta$, we deduce that
\begin{equation}
\mathbb{E}(2\log|p_{N}(x)|) = N(2x^{2}-1-2\log(2))+C'(0) + O(\log(N)/N).
\end{equation}

The formula (\ref{krasov}) suggests that, apart from the factors $C(\delta_{j})$  which as we shall see play no role in our calculations, the faithful description of $2\log|p_{N}(x_{j})|$ is that of the regularized GLC process with covariance (\ref{covop}), the {\it position-dependent} variance $2\left(\ln{N}+\ln{\sqrt{1-x^2}}-\ln{2}\right)$ and the {\it position-dependent} mean $N(2x^2-1-2\log(2))$.
  We find it convenient to subtract the mean value and concentrate on the centered GLC $\phi_N(x)$ in Prediction 1. As to the position-dependent logarithmic variance (stemming from the factors $(1-x_{j}^{2})^{\delta_{j}^{2}/2}$ in (\ref{krasov})) we shall see that it does play a very essential role in statistics of the position of global maximum for $|p_{N}(x)|$ via giving rise to nontrivial "edge charges" in
the corresponding Jacobi ensemble. This observation corroborates with the earlier mentioned fact that the subleading position-dependent term $f(x)$ in the variance of the LCG, see (\ref{LCdefa}), may modify the extreme value statistics.

\subsubsection{Log-correlated Gaussian random potential { (LCGP)} with a background potential:}

\noindent

\medskip

An interesting question is to study the position of the {\it minimum} for the sum of a LCG { random potential and
of a determistic background potential,} i.e:
\bea
x_m = {\rm Argmin}_{x \in D} \big( V(x) + V_0(x) \big)
\eea
Here we obtain results when $D$ is an interval, say $x \in [0,1]$.
The LCG { random potential} has correlations:
 \bea
 \mathbb{E}\{V(x) V(x')\} = C_\epsilon(x-x') \quad , \quad \lim_{\epsilon \to 0} C_\epsilon(x) = - 2 \ln|x| \quad |x|>0
 \label{C}
 \eea
and $C_\epsilon(0)=2 \ln(1/\epsilon)$
where $\epsilon$ is a small scale regularization. The background potential is of the special
logarithmic form: 
\bea \label{V0}
V_0(x) = - \bar a \ln x -  \bar b \ln(1-x)   
\eea
{ which we will often refer to, following the Coulomb gas language, as "edge charges" }at the boundary.
We mainly focus on the case of repelling charges, $\bar a, \bar b>0$, although both the
{ model, and some of our results, extend to some range of attractive charges.}
Some properties of this model, such as the { PDF of the value of the total potential at the
minimum,
were studied in \cite{FLDR09}. Here we obtain, for the two lowest moments
of $x_m$} \\

\noindent{\bf Prediction 2.}
{\it
\begin{eqnarray} \label{p3-1}
\mathbb{E} \left\{ x_m \right\} - \frac{1}{2} = \frac{\bar a- \bar b}{2 (\bar a+\bar b+4)}
\eea
\bea  \label{p3-2}
\mathbb{E}\left\{ x_m^2 \right\}
- (\mathbb{E} \left\{ x_m \right\})^2 = \frac{(\bar a+2) (\bar b+2) (2 \bar a+2 \bar b+9)}{(\bar a+\bar b+4)^2
(\bar a+\bar b+5)^2}
\eea
}
Note that for the background potential $V_0(x)$ alone, i.e. in the absence of disorder,  and for $\bar a, \bar b >0$, the minimum  for the background potential $V_0(x)$ alone is at $x_m^0=\frac{\bar a}{\bar a + \bar b}$, that is $x_m^0 - \frac{1}{2} = \frac{\bar a- \bar b}{2 (\bar a+\bar b)}$. Hence the disorder brings
the minumum closer in average to the midpoint $x=\frac{1}{2}$.

\subsubsection{Fractional Brownian motion with Hurst index $H=0$ { (fBm0)}:}

\noindent

\medskip

The {\it fractional Brownian motion}
introduced by Kolmogorov in 1940 and
 rediscovered in the seminal work by Mandelbrot $\&$ van Ness \cite{ManvNess68} is defined as the Gaussian process  with zero mean and with the covariance structure:
\begin{equation}\label{intr1}
\mathbb{E}\left\{B_H(x_1)B_H(x_2)\right\}=\frac{\sigma_H^2}{2}\left(|x_1|^{2H}+|x_2|^{2H}-|x_1-x_2|^{2H}\right),
\end{equation}
where $0<H<1$ and $\sigma_H^2=\mbox{Var}\{B_H(1)\}$. The utility of these long-ranged correlated processes is related to the properties of being {\it self-similar} and
having {\it stationary increments}, which characterize the corresponding family of Gaussian process uniquely. In particular, for $H=1/2$ the fBm $B_{1/2}(t)\equiv B(t)$ is the usual Brownian motion (Wiener process).
Note however that naively putting $H=0$ in (\ref{intr1}) does not yield  a  well-defined process.
Nevertheless we will see below that the limit $H\to 0$ for fractional Brownian motion  can be properly defined after appropriate regularization and yields a Gaussian process with logarithmic correlations.

Consider a family of Gaussian processes depending on two parameters: $0\le H<1$ and a regularization
$\eta>0$ and given explicitly by the integral representation \cite{FKS13}
\begin{equation}\label{intr3}
 \fl ~~~ B^{(\eta)}_H(x) = \frac{1}{\sqrt{2}}\int_0^{\infty}\frac{e^{-\eta s}}{s^{1/2+H}}\left\{\left[e^{-ixs}-1\right]B_c(ds)/2+
 \left[e^{ixs}-1\right]\overline{B_c(ds)}/2\right\}.
\end{equation}
Here $B_c(s)=B_R(s)+iB_I(s)$, with $B_R(s)$ and $B_I(s)$ being two independent copies of the Wiener process $B(t)$ (the standard Brownian motion) so that $B(dt)$ is the corresponding white noise measure, $\mathbb{E}\left\{B(dt)\right\}=0$ and $\mathbb{E}\left\{B(dt)B(dt')\right\}=\delta(t-t')dtdt'$.

 The regularized process $\{B^{(\eta)}_H(x): x\in \mathbb{R} \}$ is Gaussian, has zero mean and is characterized by the covariance structure
\begin{equation}\label{intr4}
\mathbb{E}\left\{B^{(\eta)}_H(x_1)B^{(\eta)}_H(x_2)\right\}=\phi^{(\eta)}_H(x_1)+\phi^{(\eta)}_H(x_2)-\phi^{(\eta)}_H(x_1-x_2),
\end{equation}
where
\begin{eqnarray}
\phi^{(\eta)}_H(x)&=\frac{1}{2}\int_0^{\infty}\frac{e^{-2\eta s}}{s^{1+2H}}\left(1-\cos{(xs)}\right)\, ds\\
&=\frac{1}{4H}\Gamma(1-2H)\left[(4\eta^2+x^2)^H\cos{\left(2 H \arctan{\frac{x}{2\eta}}\right)}-(2\eta)^{2H}\right].  \label{intr5}
\end{eqnarray}
It is easy to verify that for any $0<H<1$ one has $\lim_{\eta \to 0}B^{(\eta )}_{H}(x)=B_{H}(x)$
 which is precisely the fBm defined in (\ref{intr1}).

As has been already mentioned the limit $\eta \to 0$ for $H=0$ does not yield any well-defined process. At the same time taking the limit $H \to 0$ at fixed $\eta$ gives
\begin{equation}\label{intr6}
\lim_{H\to 0}\phi^{(\eta )}_H(x)=\frac{1}{4}\log{\frac{x^2+4\eta^2}{4\eta^2}},
\end{equation}
ensuring that for any $\eta>0$ the limit of $B^{(\eta)}_{H}(x)$ as $H\to 0$ yields a well-defined Gaussian process $\{ B^{(\eta)}_{0}(x): x\in \mathbb{R} \}$
with stationary increments and with the increment structure function depending {\it logarithmically} on the time separation:
\begin{equation}\label{intr7}
\mathbb{E}\left\{\left[B^{(\eta )}_0(x_1)-B^{(\eta )}_0(x_2)\right]^2\right\}=\frac{1}{2}\log{\frac{|x_1-x_2|^2+4\eta^2}{4\eta^2}}\,.
\end{equation}
 We consider $B^{(\eta )}_{0}(x)$ as the most natural extension of the standard fBm to the case of zero Hurst index $H=0$. We will frequently refer to this process as fBm0. The process is regularized at scales $|x_1-x_2|<2\eta$.

 It is also worth pointing out that there exists an intimate relation between $B^{(\eta )}_0(x)$ and the behaviour of the ({\it increments} of)  GUE characteristic polynomials, though at a different, so-called "mesoscopic" spectral scales \cite{FKS13}, negligible in comparision with the interval $[-1,1]$.  The mesoscopic intervals are defined as those typically containing in the limit $N\to \infty$ a number of eigenvalues growing with $N$, but representing still a vanishingly small fraction of the total number $N$ of all eigenvalues.
  In other words, fBM0 describes behaviour of the (logarithm of) the {\it ratio} of the moduli of characteristic polynomial at mesoscopic difference in spectral parameter. More precisely, $B^{(\eta )}_0(x)$ is , in a suitable sense, given by $N\to \infty$ limit of the following object:
\begin{equation}
\label{wntau}
 W_{N}(x) = \frac{1}{2\pi}\left(-\log\det\left[\left(\frac{x}{d_{N}}I-H\right)^{2}+\frac{\eta^{2}}{4d_{N}^{2}}\right]+
 \log\det\left[H^{2}+\frac{\eta^{2}}{4d_{N}^{2}}\right]\right)
\end{equation}
where $H$ is an $N \times N$ random GUE matrix, parameter $\eta>0$ is a regularisation ensuring that the logarithms are well defined for real $x$ and $d_{N}$ specifies the asymptotic scale of the spectral axis of $H$ as $N \to \infty$ and is chosen to be {\it mesoscopic} $1 \ll d_{N} \ll N$ (say, $d_{N} = N^{\gamma}$ with $0 < \gamma < 1$).

Applying our methods of dealing with LCG to fbM0 yields the following predictions
for a few lowest moments of the position of the global minimum for the process $B^{(\eta )}_{0}(x)$
in the interval $0,L$:\\[1ex]
\noindent{\bf Prediction 3.} {\it
Define $V(x)=2B^{(\eta )}_{0}(x)$ and for $\eta>0$ and $L>0$ consider the random variable
\begin{equation}
y_{m}(\eta,L) :=  \frac{1}{L}\mbox{Arg min}_{x\in [0,L]} V(x) \label{maxrv}
\end{equation}
Then the lowest even integer moments of this random variable have the values
\begin{eqnarray}
\lim_{\eta \to 0}\mathbb{E}\left\{\left[y_m(\eta,L)\right]^2\right\}= \frac{17}{50},
\quad \lim_{\eta \to 0}\mathbb{E}\left\{\left[y_m(\eta,L)\right]^4\right\}= \frac{311}{1470}
\end{eqnarray}
whereas the odd integer moments can be found from the identities
\begin{eqnarray}
\lim_{\eta \to 0}\mathbb{E}\left\{\left[y_m(\eta,L)-\frac{1}{2}\right]^{2k+1}\right\}=0, \quad k=0,1,2,\ldots
\end{eqnarray}
}

 One should point out an interesting difference in application of our method to this case,
concerning the {\it value} of the minimum $V_m = \min_{x \in [0,1]} V(x)$,
which seems to be a direct consequence of non-stationarity of the fBm0. Since the process
is constrained to the value $V(0)=2B^{(\eta )}_{0}(0)=0$ at zero,
$V_m$ is necessarily negative or zero, at variance with the other cases
studied here. As discussed in Appendix G, this implies that the method of analytical
continuation in $n$ of Ref. \cite{FLDR09}, which works nicely for the other cases,
fails to predict the PDF of $V_m$ for the fBm0,
and requires modifications which are left for future studies.
We do not believe that this problem bears consequence to the {\it moments of $x_m$},
which enjoy a nice and simple analytical continuation to $n=0$. As we checked
numerically up to large $n$, these moments pass the standard tests (i.e. positivity of Hankel matrices)
for existence of a positive associated PDF. {\it Conditional moments} however, i.e. conditioned to an atypically
high value of $V_m$, would need a more careful study, beyond the
scope of this paper.

\subsection{Moments of the eigenvalue density of the Jacobi ensemble}
\label{sec:mom}

As mentioned above the statistics of extrema in log-correlated fields on an interval relate to the $\beta$-Jacobi ensemble of random matrices. We denote ${\bf y}=(y_1,\ldots,y_n)$ the set of
eigenvalues, with $y_i \in [0,1]$, $i=1,\ldots,n$. The model can be defined \cite{DE2002,Forbook}
by the joint distribution of eigenvalues
\footnote{Note that we use $2 \kappa$ for the Dyson index instead of $\beta$ to avoid confusion with
the inverse temperature, also denoted $\beta$, associated to the study of the log-correlated process.}
\be \label{jpdf}
\fl  {\cal P}_J({\bf y})  d {\bf y} = \frac{1}{{\cal Z}_n} \prod_{i=1}^n dy_i y_i^a (1-y_i)^b |\Delta({\bf y})|^{2 \kappa}   \quad , \quad \Delta({\bf y}) = \prod_{1 \leq i < j \leq n} (y_i-y_j)
\ee
where the normalization constant, ${\cal Z}_n$, is the famous Selberg integral
\cite{ForWar} for which an explicit formula exists for any positive integer $n$

\begin{equation}\label{MKIb}
{\cal Z}_n = Sl_n(\kappa,a,b):=\int_{[0,1]^n}|\Delta({\bf y)}|^{2\kappa}\prod_{i=1}^n y_i^{a}(1-y_i)^{b}dy_i
\end{equation}
\[
=\prod_{j=0}^{n-1}\frac{\Gamma\left(a+1+\kappa j\right)\Gamma\left(b+1+\kappa j\right)\Gamma\left(1+\kappa (j+1)\right)}
{\Gamma\left(a+b+2+\kappa (n+j-1)\right)\Gamma\left(1+\kappa\right)}
\]
In \cite{FLDR09} we analytically continued this formula to complex $n$
which allowed to obtain
the probability distribution of the height of the global minimum of the process
(see also \cite{Ostrovsky1},\cite{Ostrovsky2},\cite{Ostrovsky3},\cite{Ostrovsky4}).
In the present paper we advance this analysis much further in order to extract the statistics of the {\it position} of the global minimum. As we will show this requires to obtain some exact formula for the moments of the
eigenvalue density for the $\beta$-Jacobi ensemble for arbitrary positive integer $n$. Furthermore these
should be explicit enough to allow for a continuation to $n=0$.

Let us define the average value $<f(y)>_J$ of any function $f(y)$ over the Jacobi density by the relation
\begin{equation} \label{av}
<f({\bf y})>_J:=[Sl_n(\kappa,a,b)]^{-1}\int_{[0,1]^n} \, f({\bf y})\,|\Delta({\bf y})|^{2\kappa}\prod_{i=1}^n y_i^{a}(1-y_i)^{b}dy_i
\end{equation}
In particular for any positive integer $n$ the mean density of eigenvalues is defined as:
\bea
\rho_J(y) = < \frac{1}{n} \sum_{i=1}^n \delta(y_i-y) >_J = <\delta(y_1-y)>_J
\eea
where $\delta$ is the Dirac distribution. The moments studied here are then defined as:
\bea
M^{(J)}_k = \int_0^1 dy  y^k \rho_J(y)  = < y_1^k >_J
\eea

The problem of calculating these moments, for the $\beta$-Jacobi ensemble, has been already addressed in the theoretical physics and
mathematical literature, motivated by various applications.
In full generality it turns out to be a hard problem, and only limited results were available.
One method is based on recursion on the order of the moment, as outlined in the chapter 17 of Mehta's book \cite{Meh04}  and used in \cite{SS2006,SS2008} in the context of  conductance distribution in chaotic transport through mesoscopic cavities. In that approach higher moments calculations become technically unsurmountable.
Another approach using Schuhr functions was developed and gave very explicit results for these and
other moments for $\beta=1,2$ \cite{SS2009} (see also related work in \cite{Vivo} and
\cite{MS2013}).

More recently, an interesting contour integral representation for those moments was proved by Borodin and Gorin, but remained unpublished. It is described in the Appendix A to the present paper, provided by these authors.
It allows a systematic calculation of the moments, including negative ones. However evaluating these
integrals becomes again a challenge for higher moments.

In the present paper we give an explicit expression  for all integer moments, positive and negative, of the eigenvalue density
for the $\beta$-Jacobi ensemble, based on a different approach, in terms of sums over partitions.

Let us further denote
$\lambda=(\lambda_1 \geq \lambda_2 \geq  .. \geq \lambda_{\ell(\lambda)})$ a partition of length $\ell(\lambda)$
of the integer $k \geq 0$,
with $\lambda_i$ strictly positive integers such that $|\lambda|=\sum_{i=1}^{\ell(\lambda)} \lambda_i=k$. Then
 we obtain
\be \label{momentsJ1}
 \left\langle \frac{1}{n} \sum_{j=1}^n y_j^k \right\rangle_J = \sum_{\lambda, |\lambda|=k} A_\lambda
 a^+_\lambda  \quad , \quad
  \left\langle \frac{1}{n} \sum_{j=1}^n y_j^{-k} \right\rangle_J = \sum_{\lambda, |\lambda|=k} A_\lambda
  a^-_\lambda
\ee
where the sum is over all partitions of $k$, and
\bea \label{momentsJ2}
\fl && A_\lambda = \frac{ k (\lambda_1-1)! }{
(\kappa(\ell(\lambda)-1)+1)_{\lambda_1}}
\prod_{i=2}^{\ell(\lambda)}
 \frac{(\kappa(1-i))_{\lambda_i}}{(\kappa(\ell(\lambda)-i)+1)_{\lambda_i}}  \prod_{1 \leq i < j \leq \ell(\lambda)}
\frac{\kappa(j-i)+\lambda_i-\lambda_j}{\kappa(j-i) }   \\
\fl && \times \frac{1}{n} \prod_{i=1}^{\ell(\lambda)}
  \frac{(\kappa(n-i+1))_{\lambda_i} }{ (\kappa(\ell(\lambda)-i+1))_{\lambda_i} }
   \prod_{1 \leq i < j \leq \ell(\lambda)} \frac{(\kappa(j-i+1))_{\lambda_i-\lambda_j} }{ (\kappa(j-i-1)+1)_{\lambda_i-\lambda_j}  }
   \nonumber
\eea
in terms of the Pochhammer symbol $(x)_n = x (x+1) .. (x+n-1) = \Gamma(x+n)/\Gamma(x)$,
with for positive moments
\bea \label{momentsJ2pos}
&& a^+_\lambda = \prod_{i=1}^{\ell(\lambda)}  \frac{(a+1+ \kappa(n-i))_{\lambda_i} }{
  (a+b+2+\kappa (2 n - i -1))_{\lambda_i} }
\eea
while for negative moments
\bea \label{momentsJ2neg}
&& a^-_\lambda = \prod_{i=1}^{\ell(\lambda)}  \frac{(a+1+ \kappa(i-1))_{-\lambda_i} }{
  (a+b+2+\kappa (n + i -2))_{- \lambda_i} }
\eea
This result as it stands was derived for $n, k \in \mathbb{N}$
and $\kappa >0$, and in range of values of $a,b$ such that these moments exist.
Later however we will study analytic continuations in these
parameters.

For positive moments, this very explicit formula is equivalent to a result by Mezzadri and Reynolds, which appeared in
\cite{MR15} during the course of the present work. Our independent derivation is relatively straightforward and will be given below for both positive and negative moments.

One can check on the above expressions (\ref{momentsJ1})-(\ref{momentsJ2}),
that in the limit $\kappa \to 0$ partitions contribute as $\sim \kappa^{\ell(\lambda)-1}$, hence only the single partition $\lambda=(k)$ with a single row contributes, leading to the trivial limit for the moments:
\bea \label{zerokappa}
\left\langle \frac{1}{n} \sum_{j=1}^n y_j^k \right\rangle_{J, \kappa=0} = \frac{(a+1)_k}{(a+b+2)_k}
\eea
 here for $k$ of either sign,
as expected, since in that limit the Jacobi measures decouples
${\cal P}_J({\bf y}) = \prod_{i=1}^n P^0(y_i)$ where $P^0(y)=
\frac{\Gamma (a+b+2)}{\Gamma (a+1) \Gamma (b+1)} y^a (1-y)^b$. Other general properties of
the moments, such as identity of moments of $y$ and of $1-y$ for $a=b$, are
less straightforward to see on the above formula, and are explicitly checked below for low moments.

\subsection{Replica method and freezing duality conjecture for log-correlated processes}
Our method of addressing statistics related to the global {\it minimum} of random functions (which can be trivially
adapted for the global maximum with obvious modifications)
 is inspired by statistical mechanics of disordered systems.  Namely, we look at any random function $V(x)$ defined in an interval $D\in \mathbb{R}$ of the real axis  as a one-dimensional random potential, with $x$ playing the role of spatial coordinate. To that end we introduce for any $\beta>0$ and any positive $\beta-$independent weight function $\mu(x)>0$
  the associated Boltzmann-Gibbs-like equilibrium measure by
 \be\label{2}
 p_{\beta}(x)=\frac{1}{Z_{\beta}}\,\mu(x) e^{-\beta V(x)}\equiv \frac{1}{Z_{\beta}}\,
\int_{D}\,\delta(x-x_1)\, e^{-\beta V(x_1)} \mu(x_1)\,dx_1,
\ee
where we have defined the associated normalization function (the "partition function")
 \be\label{3}
 Z_{\beta}=\int_{D}\,e^{-\beta V(x)} \mu(x)\,dx
\ee
with $\beta=1/T$ playing the role of inverse temperature.

According to the basic principles of statistical mechanics in the limit of zero temperature $\beta\to \infty$ the Boltzmann-Gibbs measure must be dominated by the minimum of the random potential $V_m=\min_{x\in D}\,V(x)$ achieved at the point $x_m\in D$.
The latter is randomly fluctuating from one realization of the potential to another. The probability density for the position of the minimum is defined as
${\cal P}(x)=\overline{\delta(x-x_m)}$ where from now on we use the bar to denote the expectation with respect to random process ("potential disorder")  realizations: $$\overline{\left(\ldots\right) }\equiv \mathbb{E}\left\{\left(\ldots \right)\right\} $$
This leads to the fundamental relation
\be \label{4}
~~~~ {\cal P} (x)=\lim_{\beta\to \infty}\overline{ p_{\beta}(x)}
\ee
Therefore, calculating ${\cal P}(x)$ amounts to (i) performing the disorder average of the Boltzmann-Gibbs measure (\ref{2}) and (ii) evaluating its zero-temperature limit. The second step is highly non-trivial,
due to a phase transition occuring at some finite value $\beta=\beta_c$. In our previous work on
decaying Burgers turbulence \cite{FLDR10}, which turns out to be a limiting case of the present
problem (see Section \ref{sec:gauss}), we have already succeeded in implementing that program.
Following the same strategy, the step (i) is done using the
replica method, a powerful (albeit not yet mathematically rigorous) heuristic method of theoretical
physics of disordered systems. It amounts to
representing $Z_{\beta}^{-1}=\lim_{n\to 0} Z_{\beta}^{n-1}$ which after assuming integer $n>1$ results in the formal identity
\begin{equation}
 \overline{p_\beta(x)} = \lim_{n \to 0} p_{\beta,n}(x)
\end{equation}
where
\bea \label{replica}
\fl && p_{\beta,n}(x) = \mu(x) \overline{ e^{-\beta V(x)}  Z_\beta^{n-1}} = \int_{x_1
\in D} ... \int_{x_n \in D} ~ \overline{e^{- \beta \sum_{=1}^n V(x_i)} } ~
\delta(x-x_1)\prod_{i=1}^n \mu(x_i)\,  dx_i
\eea
Note that $p_{\beta,n}(x)$ is not a probability distribution for general $n$, but becomes one for $n=0$.

In the next section we will show how to calculate the moments $M_k(\beta)=\int_{D} p_{\beta,n}(x) x^k \,dx$
for positive integer $k$ and any integer value $n>0$ in the high temperature phase of the model $\beta<\beta_{c}$. Note that by a trivial rescaling of the potential one always can ensure $\beta_c=1$,
setting $g=1$ in (\ref{LCdefa}), and we assume such a rescaling henceforth.
In the range $\beta<1$ the formulae we obtain for the moments turn out to be easy to
continue to $n=0$. This yields
the integer moments of the probability density $\overline{p_\beta(x)}$ in that
phase. There still however remains the task of finding a way to continue those expressions to $\beta>1$ in order to compute the limit $\beta\to \infty$ and extract the information about the Argmin distribution ${\cal P} (x)$.
To perform the continuation, we rely on the {\it freezing transition scenario} for logarithmically correlated random landscapes. The background idea
goes back to \cite{CLD01} and was advanced further in \cite{FB08} leading to explicit predictions.
In \cite{FLDR09} it was discovered that the duality property appears to play a crucial role, leading
to the {\it freezing-duality conjecture} (FDC), which was further utilized in
\cite{FLDR10,FLDR12,FHK12,FK14,FyoSim15,CRS}.
In brief, the FDC predicts a phase transition at the critical value $\beta=1$
and amounts to the following principle:
 \begin{center}
 {\it Thermodynamic quantities which for $\beta<1$ are {\bf duality-invariant} functions of the inverse temperature $\beta$, that is remain invariant under the transformation $\beta\to \beta^{-1}$, "{\bf freeze}" in the low temperature phase, that is retain for all $\beta>1$  the value they acquired at the point of self-duality $\beta =1$.}
 \end{center}

Here, on our explicit formula, we will indeed be able to verify that every integer moment $M_k(\beta)=\int_{D} p_{\beta,n=0}(x) x^k \,dx$ of the probability density $\overline{p_\beta(x)}$ is {\it
duality-invariant} in the above sense, and hence can be continued to
$\beta>1$ using the FDC, yielding the moments of the position of the global minimum.
Another way of proof is based on the powerful contour integral representation for the moments of Jacobi ensemble of random matrices provided by Borodin and Gorin. We thus conjecture that not only all moments, but the whole disorder averaged Gibbs measure $\overline{p_\beta(x)}$ freezes at $\beta=1$, hence that the PDF of the position of the minimum is determined as
\bea
{\cal P}(x) = \lim_{\beta \to 1} \overline{p_\beta(x)}
\eea
similar to the conjecture in \cite{FLDR10} in our study of the Burgers equation.

Although the FDC scenario is not yet proven mathematically in full generality and has a status of a conjecture supported by physical arguments and available numerics, recently a few nontrivial aspects of freezing were verified within rigorous probabilistic analysis, see e.g. \cite{AZ14,SZ2014,vargasfreezing,DRZ2015} for
 progress in that direction
  \footnote{In \cite{CLD01,FB08,FLDR09} we predicted the freezing of the generating
  function $g_\beta(y):=\overline{\exp(-e^{\beta y} z_\beta)}$, where
  $z_\beta$ is a regularized version of the partition sum $Z_\beta$.
  It was proved
  in \cite{vargasfreezing} (see corollary 2.3) in a more general
  and rigorous setting. It tells about the
  free energy $f_\beta=-\beta^{-1} \ln z_\beta$ and the {\it value} of
  the minimum $f_\infty$. Indeed, by construction $1-g_\beta(y)$ is the cumulative distribution function
  of the random variable defined as $y_\beta:=f_\beta - \beta^{-1} G$ where
  $G$ is a unit Gumbell random variable, independent from $f_\beta$. The PDF of
  $y_\beta$ for any $\beta<\beta_c$, of $y_{\beta_c}$ and of $f_\infty$ are thus
  identical. See the Appendix H for more details. }. However the role of duality has not yet been verified rigorously. Interesting connections to
  duality in Liouville and conformal field theory \cite{Zamo1996} remain to be clarified.

In the rest of the paper we give a detailed derivation of the outlined steps of our procedure
and an analysis of the results.

\section{Acknowledgements}
The authors are very grateful to Nick Simm for kindly providing numerical data on argmax of GUE, as well as to Alexei Borodin and Vadim Gorin for bringing their methods to our attention, for writing the Appendix A in the
present paper, and for indicating relevant references to us. We would like to warmly thank the anonymous referee for suggesting to use formulas (\ref{res2neg},\ref{invert}) for obtaining an explicit expression for the negative moments, Alberto Rosso for a lively discussions at the early stage of the project,
and Dima Savin for guiding us to the literature on the moments of Jacobi density. We also
thank F. Mezzadri and A. Reynolds for informing us on their moment formulae prior to publication.
Kind hospitality of the Newton Institute, Cambridge during the program "Random Geometry"
as well as of the Simon Center for Geometry and Physics in Stony Brook, where this research
 was completed, is acknowledged with thanks. YF was supported by EPSRC grants EP/J002763/1 ``Insights into Disordered Landscapes via Random Matrix Theory and Statistical Mechanics''
and EP/N009436/1 "The many faces of random characteristic polynomials".
PLD was supported by PSL grant ANR-10-IDEX-0001- 02-PSL.

\section{Calculations within the Replica Method}

Our goal in this section is to develop the method of evaluating the required disorder average and calculating the resulting integrals explicitly in some range of inverse temperatures $\beta$ for a few instances of the log-correlated random potentials $V(x)$.

\subsection{Connections to the $\beta$-Jacobi ensemble}

\subsubsection{GUE characteristic polynomial.}

In that case we will follow the related earlier study in \cite{FyoSim15} and use the family of weight factors
$\mu(x) = \rho(x)^q$ on the interval $x \in [-1,1]$, with $\rho(x)=\frac{2}{\pi} \sqrt{1-x^2}$ and parameter $q>0$.
Such a choice of the weight is justified {\it a posteriori} by the possibility to find within this family a duality-invariant expression for the moments in the high-temperature phase which is central for our method to work.
Since here we are interested in the {\it maximum} of the characteristic polynomial we will define the potential $V(x)=- \phi_{N}(x)$ where $\phi_{N}(x)$, defined in {\bf Prediction 1}, is not strictly a Gaussian field. Nevertheless due to
the Krasovsky formula (\ref{krasov})  we can write asymptotically in the limit of large $N\gg 1$
\bea
&&  \overline{ e^{- \beta \sum_{a=1}^n V(x_a)}\, } =
\overline{ e^{\beta \sum_{a=1}^n \phi_{N}(x_a)}\, }  = \overline{ \prod_{a=1}^n |p_N(x_a)|^{2 \beta} }
e^{-2 \beta \sum_{a=1}^n \overline{\ln |p_{N}(x_a)|}}\, \\
&& \simeq A_n \prod_{a=1}^n (1-x_a^2)^{\beta^2/2}
\prod_{a<b} |x_a-x_b|^{-2 \beta^2}\label{case2}
\eea
with $A_n = [ C(\beta) (N/2)^{\beta^2} e^{- \beta C'(0)} 2^{- \beta^2 (n-1)} ]^n$.

\subsubsection{General scaled model in $[0,1]$.}

It is easy to see that by proper rescaling the evaluation of the function $p_{\beta,n}(x)$ from (\ref{replica}) amounts in the case of both characteristic GUE polynomials and the LCGP process to evaluating particular cases of the following multiple integral defined on the interval $[0,1]$:
\begin{eqnarray} \label{rescaled}
\fl p_{\beta,a,b,n}(y) = \int_0^1 .. \int_0^1 \prod_{i=1}^n dy_i y_i^{a} (1-y_i)^{b}
 \prod_{1 \leq i < j \leq j \leq n} \frac{1}{|y_i -y_j|^{2 \beta^2}} ~ \delta(y-y_1)
\end{eqnarray}
which is formally the (un-normalized) density of eigenvalues of the
$\beta$-Jacobi ensemble
(\ref{jpdf}), i.e. $p_{\beta,a,b,n}(y) = {\cal Z}_n  \rho_J(y)$,
however with a negative value of the parameter $\kappa=- \beta^2$.
In particular the normalisation factor ${\cal Z}_n$ is the Selberg integral (\ref{MKIb}). Note that both integrals are well defined for
$\beta^2 < 1$ and positive integer $n$. They can also be defined for larger $\beta$, upon
introducing an implicit
small scale cutoff which modifies the expressions for $|y_i-y_j|<\epsilon$.
However we will use only
the high temperature regime and will continue analytically our moment formula to
$n=0$ and $\beta =1$.

These integrals are associated to
the statistical mechanics of a random energy model generated
by a LCG field in the interval $[0,1]$ in presence of boundary charges
of strength $a$ and $b$. More precisely, they appear in the
study of the continuum partition function introduced in \cite{FLDR09}
\bea \label{modelab}
Z_\beta= \epsilon^{\beta^2} \int_0^1 dy y^a (1-y)^b e^{-\beta V(y)}
\eea
where the correlator of $V(x)$ was defined in (\ref{C}). Here $a$ and $b$ are the
two parameters of the model and the factor $\epsilon^{\beta^2}$
ensures that the integer moments $\overline{Z_\beta^n}$
are $\epsilon$-independent in the high temperature regime.
There, these moments $\overline{Z_\beta^n} = {\cal Z}_n$ are given by the Selberg integral (\ref{MKIb}).
Several aspects of this model, such as freezing, duality and obtaining the
PDF of the value of the minimum, were analyzed in \cite{FLDR09}.
Here we will focus instead on the calculation of the moments
of the position $y$. In each realization of the random potential $V$
they are defined as
\bea \label{momentsab}
< y^k>_{\beta,a,b} = \frac{\int_0^1 dy y^k y^a (1-y)^b e^{-\beta V(y)}}{\int_0^1 dy y^a (1-y)^b e^{-\beta V(y)}}
\eea
and we will be interested in calculating their disorder averages $\overline{< y^k>_{\beta,a,b}}$.

\subsubsection{Fractional Brownian motion.}

For our fBm0 example we take the weight function $\mu(x)=1$ in (\ref{2}), (\ref{3}) on the interval $x \in [0,L]$,
and use the rescaled fBm with $H=0$ as the random potential: $V(x)=2 g B^{(\eta )}_0(x)$. Exploiting the Gaussian nature of fBm0 we can easily perform the required average using (\ref{intr4}) which yields
\bea
&& \overline{ e^{-2 \beta g \sum_{i=1}^n B^{(\eta )}_0(x_i)}}=e^{2 (\beta g)^2\left\{\sum_{i=1}^n \overline{\left[B^{(\eta )}_0(x_i)\right]^2}+
2\sum_{i<j}^n \overline{ B^{(\eta )}_0(x_i) B^{(\eta )}_0(x_j)}\right\}} \\
&& =e^{n (2 \beta g)^2\sum_{i=1}^n \phi^{(\eta)}_0(x_i)-(2 \beta g)^2
 \sum_{i<j}^n \phi^{(\eta)}_0(x_i-x_j)}
\eea
Further using the definitions (\ref{intr6})-(\ref{intr7}) we arrive at
\bea
 &&
 \overline{ e^{- 2 \beta g \sum_{i=1}^n B^{(\eta )}_0(x_i)}}=
 \prod_{i=1}^n \left[\frac{x_i^2+4\eta^2}{4\eta^2}\right]^{a/2}
 \prod_{i<j}^n \left[\frac{(x_i-x_j)^2+4\eta^2}{4\eta^2}\right]^{- \gamma} \\
 &&  a= 2 n \gamma \quad , \quad \gamma = g^2 \beta^2
\eea
Assuming that all $x_i>0$ we can write in the limit of vanishing regularization $\eta\to 0$:
\begin{equation}
 \overline{ e^{-2 \beta g \sum_{i=1}^n B^{(\eta )}_0(x_i)}}\approx (2\eta)^{n(\gamma(n-1)-a)}
 \prod_{i=1}^n x_i^a
 \prod_{i<j}^n |x_i-x_j|^{-2\gamma}\label{case1}
\end{equation}
  For convenience we will use the particular value $g=1$ ensuring {\it a posteriori} the critical temperature value $\beta_c=1$.\\[1ex]

It is easy to see how our three examples can be studied
within the framework of the model (\ref{rescaled}):

\begin{itemize}

\item GUE characteristic polynomial. We define $x= 1-2 y$, , with $y \in[0,1]$ and
\bea
&& p_{\beta,n}(x) dx = C_n p_{\beta,a,a,n}(y) dy \quad , \quad a = \frac{q + \beta^2}{2} \\
&& C_n =  [ (\frac{2}{\pi})^{q}  C(\beta) N^{\beta^2} e^{- \beta C'(0)}
2^{1+ q - 2 \beta^2 (n-1)} ]^n
\eea

\item LCGP plus a background potential on interval $[0,1]$, defined in (\ref{C}) and (\ref{V0}). One sees that this
model is obtained by choosing $a = \beta \bar a$ and $b= \beta \bar b$.

\item fBm0: we define $x=L y$, with $y \in[0,1]$ and:
\bea
p_{\beta,n}(x) dx = L^{n(n+1) \gamma + n}  p_{\beta,a=2 n \beta^2,b=0,n}(y) dy
\eea

\end{itemize}

Although the explicit form of the density $\rho_J(y)$ of the $\beta$-Jacobi ensemble is
not known in a closed form for finite $n$, the formulae
(\ref{momentsJ1})-(\ref{momentsJ2}) displayed in the introduction provide an explicit expression
for its integer moments
\bea
&& <y^k>_{\beta,a,b,n} := \frac1{\cal Z}_n \int_0^1 dy y^k p_{\beta,a,b,n}(y)
\eea
for positive integers $k,n$ and $\beta^2<0$. The collection of
indices $\beta,a,b,n$ is now replacing the index $J$, and
some of them may be omitted below when no confusion is possible.

We will now derive the formulae (\ref{momentsJ1})-(\ref{momentsJ2}) using
methods based on Jack polynomials. In a second stage we will
continue them analytically to arbitrary $n$, including $n=0$, and to $0<\beta^2 \leq 1$,
to obtain moments for the random model of interest as:
\bea
\overline{<y^k>_{\beta,a,b}} = \lim_{n \to 0} <y^k>_{\beta,a,b,n}
\eea

\subsection{Derivation of the moment formula in terms of sums over partitions}

To calculate the moments, we now consider one of the most distinguished bases of symmetric polynomials,
the Jack polynomials,
named after Henry Jack. They play a central role for our study because their
average with respect to the $\beta$-Jacobi measure is explicitly known, due to
Kadell \cite{Kadell} (see below). As discussed
in the book \cite{Jackbook} (see reprint of the original article)
Jack introduced a special set of symmetric polynomials of $n$ variables ${\bf y}:=(y_1,\ldots,y_n)$ indexed by integer partitions $\lambda$ and dependent on a real parameter $\alpha$. He called them $Z(\lambda)$, which in modern notations are denoted as $J^{(\alpha)}_{\lambda}({\bf y})$ following I. Macdonald
who greatly developed the theory of such and related objects in the book \cite{Macdonald}.
Another important source of information is Stanley's paper \cite{Stanley}. In what follows we
use notations and conventions from \cite{Macdonald} and \cite{Stanley}.
\footnote{Note that other papers use different conventions. The reader is advised
to check the conventions with care.}

Let us recall the definition of a partition
$\lambda=(\lambda_1 \geq \lambda_2 \geq  .. \geq \lambda_{\ell(\lambda)}) >0)$
of the integer $k$, of length $\ell(\lambda)$, with $\lambda_i$ strictly positive integers
such that $|\lambda|:=\sum_{i=1}^{\ell(\lambda)} \lambda_i=k$. It can be written as
\bea
\lambda = \{ (i,j) \in \mathbb{Z}^2 ; 1 \leq i \leq \ell(\lambda) , 1 \leq j \leq \lambda_i \}
\eea
from which one usually draw the Young diagram representing the partition $\lambda$,
as a collection of unit area square boxes in the plane, centered at
coordinates $i$ along the (descending, southbound) vertical, and $j$ along the (eastbound) horizontal.
The dual (or conjugate) $\lambda'$ of $\lambda$ is
the partition whose Young diagram is the transpose of
$\lambda$, i.e. reflected along the (descending) diagonal $i=j$.
Hence $\lambda'_i$ is the number of $j$ such that $\lambda_j \geq i$.
If $s=(i,j)$ stands for a square in the Young diagram, one defines the "arm length"
$a_{\lambda}(s)=\lambda_i-j$ which
is equal to the number of squares to the east of square $s$ and the "leg length"
 $l_{\lambda}(s)=\lambda_j'-i$ as the number of squares
 to the south of the square $s$. One then defines the product
\be \label{cdef}
\fl  ~~~~~~ c(\lambda,\alpha,t):=\prod_{s\in \lambda} (\alpha a_{\lambda}(s)+l_{\lambda}(s)+t)
= \prod_{i=1}^{\ell(\lambda)}  \prod_{j=1}^{\lambda_i}  (\alpha (\lambda_i-j) + \lambda'_j - i + t)
\ee
which will be used later on
\footnote{ Note that $c(\lambda,\alpha,1)= \prod_{s\in \lambda} h^\lambda_*(s)$ and
$c(\lambda,\alpha,\alpha)=\prod_{s\in \lambda} h_\lambda^*(s)$ in terms of the
notations of \cite{Stanley}.}.

Given a partition
$\lambda$, one defines, in the theory of symmetric functions,
the monomial symmetric functions $m_\lambda({\bf y})$, over $n$ variables ${\bf y}= \{y_r\}$, $r=1,,n$
as
$m_\lambda({\bf y})=\sum_\sigma \prod_{r=1}^n y_{\sigma(r)}^{\lambda_i}$ where
the summation is over all non-equivalent permutations of the variables.
For example, given a partition $(211)$ of $k=4$,
$m_{(211)}({\bf y})=y_1^2 y_2 y_3 + y_1 y_2^2 y_3 + y_1 y_2 y_3^2$.
Another useful set of symmetric functions, of obvious importance for
the calculation of moments, are the power-sums:
$$p_\lambda({\bf y})= \prod_{i=1}^{\ell(\lambda)} \sum_{r=1}^n y_r^{\lambda_i}$$
which form a basis of the ring of symmetric functions.

Define the following scalar product, which depends on a real parameter $\alpha$ as
\bea
<p_\lambda , p_\mu > = \delta_{\lambda \mu} z_\lambda \alpha^{\ell(\lambda)} \quad , \quad
z_\lambda = 1^{q_1} 2^{q_2} .. q_1! q_2! ..
\eea
where $q_i=q_i(\lambda)$ is the number of rows in $\lambda$ whose length are equal to $i$. Here and below we suppress the arguments of all
symmetric polynomials except when explicitly needed, for example $p_\lambda({\bf y}) \to p_\lambda$.
The Jack functions
$J^{(\alpha)}_\lambda$ obey the following properties (i) orthogonality with respect to the above scalar
product (ii) fixed coefficient of highest degree in
the monomial basis
\footnote{see theorems 5.6 and 5.8 in \cite{Stanley} and (i) p 377 (ii) (10.13-10.22) in
\cite{Macdonald}.}
\bea
&& <J^{(\alpha)}_\lambda , J^{(\alpha)}_\mu> = c(\lambda, \alpha,1) c(\lambda,\alpha,\alpha)  \delta_{\lambda \mu}  \\
&& J^{(\alpha)}_\lambda = c(\lambda,\alpha,1) m_\lambda + \sum_{\nu < \lambda} u_{\lambda \nu} m_\nu
\eea

We can go from the basis of power sums to the Jack polynomial basis
by the following linear transformations:
\bea
J^{(\alpha)}_\lambda = \sum_\nu \theta^\lambda_\nu(\alpha) p_\nu  \quad , \quad
p_\nu = \sum_\lambda \gamma^\lambda_\nu(\alpha) J^{(\alpha)}_\lambda
\eea
where the coefficients are in general complicated. Note that the $k$-th
moment that we are interested in is precisely the Jacobi ensemble average
\bea
< \sum_{r=1}^n y_r^k >_J = < p_{(k)}({\bf y}) >_J
\eea
where $(k)$ denotes the
partition with only one row of length $k$. Hence we need only
the coefficient $\gamma^\lambda_{\nu=(k)}$. As we now show
one can express this coefficient in terms of
$\theta^\lambda_{\nu=(k)}$. Indeed, one can write in two ways
the following scalar product, first as
\bea
< J^{(\alpha)}_\lambda , p_\mu > =  \sum_\nu \theta^\lambda_\nu(\alpha) < p_\nu , p_\mu >
= \theta^\lambda_\mu(\alpha) z_\mu \alpha^{\ell(\mu)}
\eea
and also as
\bea
< J^{(\alpha)}_\lambda , p_\mu > =  \sum_\tau \gamma^\tau_\mu(\alpha) <  J^{(\alpha)}_\lambda , J^{(\alpha)}_\tau >
= \gamma^\lambda_\mu(\alpha)  c(\lambda, \alpha,1) c(\lambda,\alpha,\alpha)
\eea
Hence, comparing we obtain
\bea
&& \gamma^\lambda_\mu(\alpha)  = \frac{ \theta^\lambda_\mu(\alpha) z_\mu \alpha^{\ell(\mu)} }{ c(\lambda, \alpha,1) c(\lambda,\alpha,\alpha)}
\eea
valid for arbitrary partitions $\mu$ and $\lambda$, which we apply to $\mu=(k)$.

Now it turns out that $\theta^\lambda_{(k)}(\alpha)$ is known to
to be equal to:
\bea
\theta^\lambda_{(k)}(\alpha)  = \prod_{s - \{1,1\}} (\alpha a'_\lambda(s) - l'_\lambda(s))
\eea
 if $|\lambda|=k$ and zero otherwise,
(see p 383 Ex. 1 (b) in \cite{Macdonald} and (19) in \cite{Connection}),
where $a'_\lambda(s) = j-1$ and $l'_\lambda(s)=i-1$ are respectively
called the co-arm and co-leg lengths of the partition $\lambda$, and the product
does not include the
box $s=(1,1)$.

 {\it Positive moments:} this leads to the explicit result for the  positive $k$-th moment in terms of an average of the
Jack polynomial associated to the partition $(k)$:
\bea
&& < \sum_{r=1}^n y_r^k >_J = < p_{(k)}({\bf y}) >_J =  \sum_{\lambda, |\lambda|=k} \gamma^\lambda_{(k)}(\alpha)
< J^{(\alpha)}_{\lambda}({\bf y}) >_J  \label{res2} \\
&& \gamma^\lambda_{(k)}(\alpha) = \frac{ k \alpha }{ c(\lambda, \alpha,1) c(\lambda,\alpha,\alpha)} \prod_{s - \{1,1\}} (\alpha a'_\lambda(s) - l'_\lambda(s))  \label{res1}
\eea
(see Appendix F 
for an alternative rewriting of this formula).

The problem therefore amounts to evaluating the Jacobi average of $J^{(\alpha)}_{\lambda}({\bf y})$. Such averages where
evaluated for a general partition $\lambda$,
for a differently normalized set of Jack polynomials, denoted by Macdonald as
$P^{(\alpha)}_{\lambda}({\bf y})$, related to the $J^{(\alpha)}_{\lambda}({\bf y})$ as follows
\bea \label{relJP}
J^{(\alpha)}_{\lambda}({\bf y}) = c(\lambda,\alpha,1) P^{(\alpha)}_{\lambda}({\bf y})
\eea
Namely, as was conjectured by Macdonald in \cite{Macdonald}, and proved by Kadell \cite{Kadell},
there exists a closed form expression for  $P^{(\alpha)}_{\lambda}({\bf y})$ integrated with the
(unnormalized) Jacobi density over the hypercube
${\bf y}\in [0,1]^n$ with the correspondence
\bea
\alpha = 1/\kappa \label{alpha}
\eea
 It is given by
\begin{equation}\label{MKI} \int_{[0,1]^n}P^{(1/\kappa)}_{\lambda}({\bf y})|\Delta({\bf y})|^{2\kappa}\prod_{i=1}^n y_i^{a}(1-y_i)^{b}dy_i
\end{equation}
\[
=n!v_{\lambda}(\kappa)\prod_{i=1}^n\frac{\Gamma\left(\lambda_i+a+1+\kappa(n-i)\right)\Gamma\left(b+1+\kappa(n-i)\right)}
{\Gamma\left(\lambda_i+a+b+2+\kappa(2n-i-1)\right)}
\]
where
\begin{equation}\label{MKIa}
v_{\lambda}(\kappa)=\prod_{1\le i<j\le n}\frac{\Gamma\left(\lambda_i-\lambda_j+\kappa(j-i+1)\right)}
{\Gamma\left(\lambda_i-\lambda_j+\kappa(j-i)\right)}
\end{equation}
For the empty partition $\lambda=(0)$  we have $P^{(1/\kappa)}_{(0)}({\bf y})=1$ and the above integral reduces to the Selberg integral (\ref{MKIb}).

Recalling the definition of the Jacobi ensemble average (\ref{av})
and taking the ratio of (\ref{MKI}) to (\ref{MKIb}), we obtain the average of the
$P^{(1/\kappa)}({\bf y})$ polynomial
and from it, the average of $J^{(1/\kappa)}({\bf y})$. The calculation is detailed in the Appendix E
and the final result is simple and explicit for arbitrary partition $\lambda$
\bea \label{avJ}
\fl && ~~~~~~~~~  \left\langle J^{1/\kappa}_{\lambda}({{\bf y}})\right\rangle_J= \kappa^{-|\lambda|}
\prod_{i=1}^{\ell(\lambda)} \frac{(a+1+ \kappa(n-i))_{\lambda_i} }{
  (a+b+2+\kappa (2 n - i -1))_{\lambda_i} }  (\kappa(n-i+1))_{\lambda_i}
 \eea

Putting together Eqs. (\ref{res2}-\ref{res1}) and (\ref{avJ}) we obtain the
$k$-th moment as
\bea \label{moments2}
&& \left\langle \frac{1}{n} \sum_{r=1}^n y_r^k \right\rangle_J = \sum_{\lambda, |\lambda|=k}
\frac{ k \alpha }{  c(\lambda, \alpha,1) c(\lambda,\alpha,\alpha)} \prod_{s - \{1,1\}} (\alpha a'_\lambda(s) - l'_\lambda(s))  \\
&& \times \frac{1}{n \kappa^{k}}
\prod_{i=1}^{\ell(\lambda)} \frac{(a+1+ \kappa(n-i))_{\lambda_i} }{
  (a+b+2+\kappa (2 n - i -1))_{\lambda_i} }  (\kappa(n-i+1))_{\lambda_i}
\eea
where $\alpha$ should be replaced by $1/\kappa$ according to (\ref{alpha}).
Using the explicit expressions for the normalization constants
(\ref{c1}) and (\ref{c2}) in Appendix D 
and the above
definitions of the co-arm and co-leg,  one obtains the formula
(\ref{momentsJ1}), (\ref{momentsJ2}), (\ref{momentsJ2pos}), for the positive moments given
in Section \ref{sec:mom}. \footnote{Note that a textbook treatment of the material discussed in this section is given in \cite{Forbook}. The formula  for the positive moments does not appear there explicitly, but can be recovered by integrating (12.145) using (12.143).}.

We have also used:
\bea
\fl && \prod_{s - \{1,1\}} (\alpha a'_\lambda(s) - l'_\lambda(s)) =
\alpha^{k-1} \prod_{i=1}^{\ell(\lambda)} \prod_{j=1}^{~~ \lambda_i ~~ \prime}
(j-1-\kappa (i-1)) = \alpha^{k-1} (\lambda_1-1)!
\prod_{i=2}^{\ell(\lambda)} (-\kappa (i-1))_{\lambda_i} \nonumber \\
\fl &&
\eea
where the prime indicates that $i=j=1$ is excluded from the product.

Equivalently, we can rewrite the expression (\ref{moments2}) for the moments in a "geometric" form
which involves only products over boxes in the Young diagrams, see (\ref{y1})-(\ref{y2}) below.

 {\it Negative moments:} Negative moments can be obtained by
applying (\ref{res2}) to the inverse variable, here before averaging (with $k \geq 0$)
\bea
&& \sum_{r=1}^n y_r^{-k}  =  \sum_{\lambda, |\lambda|=k} \gamma^\lambda_{(k)}(\alpha)
J^{(\alpha)}_{\lambda}(\frac{1}{\bf y})   \label{res2neg}
\eea
where we denote $(\frac{1}{\bf y}) \equiv (\frac{1}{y_1},..\frac{1}{y_n})$. We now use the following relation
between Jack polynomials, for $n \geq \ell(\lambda)$
\bea \label{invert}
P_\lambda^{(\alpha)}(\frac{1}{\bf y}) = y_1^{-l} .. y_n^{-l} P_{(l^n-\lambda)^+}^{(\alpha)}({\bf y})
\eea
see \cite{Forbook} p. 643,
where $l$ is (a priori) any integer $l \geq \lambda_1$ and one denotes
\bea
(l^n-\lambda)^+ = \{ l,..., l - \lambda_1,.. , l- \lambda_{\ell(\lambda}) \}
\eea
the partition of length $n$. Using (\ref{relJP}) in (\ref{res2neg}),
inserting (\ref{invert}), using again (\ref{relJP}), one can now
average over the Jacobi measure as follows
\bea
\fl &&  \left\langle \sum_{r=1}^n y_r^{-k}  \right\rangle_J =  \sum_{\lambda, |\lambda|=k} \gamma^\lambda_{(k)}(\alpha) \frac{c(\lambda,\alpha,1)}{c((l^n-\lambda)^+,\alpha,1)}
\left\langle J_{(l^n-\lambda)^+}^{(\alpha)}({\bf y}) \right\rangle_{J, a \to a-l}
\frac{Sl_n(\kappa,a-l,b)}{Sl_n(\kappa,a,b)}
 \label{res2neg2}
\eea
where the average on the right is over a shifted Jacobi measure with
parameter $a-l,b,\kappa$ to account for the prefactor in (\ref{invert}).
For the same reason the ratio of Selberg integrals appear, since
it is the normalization of the Jacobi measure. We can now use the explicit
expressions (\ref{MKIb}) and (\ref{avJ}). One finds, after a
tedious calculation, similar in spirit to the one described above
for the positive moments, the formula
(\ref{momentsJ1}), (\ref{momentsJ2}) for the negative moments given
in Section \ref{sec:mom} with
\bea
\fl && a^-_\lambda = \prod_{i=1}^{\ell(\lambda)}  \frac{(a-l+1+ \kappa(i-1))_{l-\lambda_i} }{
  (a-l+b+2+\kappa (n + i -2))_{l- \lambda_i} }
   \frac{(a-l+b+2+\kappa (n + i -2))_{l} }{(a-l+1+ \kappa(i-1))_{l} }
  \eea
where a priori $l$ is an integer sufficiently large. In practice we found that for generic values of
$\kappa$ the final
result is independent of the choice of $l$ (for each partition), hence we chose
$l=0$ which leads to the simplest expression (\ref{momentsJ2neg})
given in in Section \ref{sec:mom}. \\

 We now discuss an interesting alternative and useful representation for the (positive and negative) integer moments in terms of contour integrals.

\subsection{Borodin-Gorin contour integral representation of the moments}

\subsubsection{Positive moments.}

Recently Borodin and Gorin proved integral representations for the
moments of the Jacobi $\beta$-ensemble. We refer to the Appendix A
for derivation and present here a summary of results in our notation.
They call these moments $\frac{1}{N} M_k(\theta,N,M,\alpha)$
and the correspondence from their four parameters to ours is:
\bea
\theta \to - \beta^2 \quad N \to n \quad \theta \alpha -1 \to a \quad , \quad  \theta (M - N +1) -1 \to b
\eea
which leads to the correspondence
\bea \label{corr1}
&& <y^k>_{\beta,a,b,n}
= \frac{1}{n} M_k(- \beta^2,n,n - 1 - \frac{1+b}{\beta^2}  , - \frac{1+a}{\beta^2}  )
\eea
Translated in our parameters
their moment formula for positive integer $n$ reads:
\bea  \label{contour1}
\fl && < y^k>_{\beta,a,b,n}
=  \frac{1}{n \beta^2}
\int \prod_{i=1}^k \frac{du_i}{2 i \pi}
 \prod_{1 \leq i < j \leq k} \frac{u_j-u_i}{(u_j-u_i + \beta^2)(u_j-u_i +1)} \\
\fl  && \times
\prod_{1 \leq i +1 < j \leq k} (u_j-u_i + 1 + \beta^2)
\prod_{i=1}^k \frac{u_i+\beta^2}{u_i + \beta^2 (1-n)} \times
\frac{u_i - 1- a }{u_i - 2 - a - b - \beta^2(1-n)}  \label{contour}
\eea
which must be supplemented with the conditions,
let us call them $C_1$:  $|u_1| \ll |u_2| \ll |u_3| ..\ll |u_k|$ and $C_2$: all the contours enclose the singularities at
$u_i = - (1-n) b^2$ and not at $u_i = 2 + a + b + \beta^2 (1-n)$.
These conditions imply that one can first perform the integral
on $u_1$ and only around the pole $u_1 = - (1-n) \beta^2$
and then iteratively on $u_2,u_3,..$ and so on.

We will investigate below the properties of this representation in
the context of the models we study here.

\subsubsection{Negative moments.}

Remarkably, Borodin and Gorin also proved a contour integral formula for negative moments, whenever they exist.
In our notations, for $k \geq 1$, it reads
\bea
\fl && < y^{-k}>_{\beta,a,b,n}
= \frac{-1}{n \beta^2}
\int \prod_{i=1}^k \frac{du_i}{2 i \pi}
 \prod_{1 \leq i < j \leq k} \frac{u_j-u_i}{(u_j-u_i - \beta^2)(u_j-u_i -1)} \nonumber  \\
\fl  && \times
\prod_{1 \leq i +1 < j \leq k} (u_j-u_i - 1 - \beta^2)
\prod_{i=1}^k \frac{u_i-n \beta^2}{u_i } \times
\frac{u_i -1 - a - b - \beta^2(1-n)}{u_i-a} \nonumber   \\
\fl && \label{contourneg}
\eea
which must be supplemented with the conditions,
(i) $C_1$:  $|u_1| \ll |u_2| \ll |u_3| ..\ll |u_k|$ and (ii) $C_2$:
all the contours enclose the singularities at $0$ and not at $a$.

\subsection{duality}

Let us now discuss an important property of these
moment, the duality.

\subsubsection{Statement of the duality on the moments.}

For clarity let us first recall the property of duality-invariance
unveiled in \cite{FLDR09}. Consider first a thermodynamic
quantity $O_\beta$ obtained in the replica formalism in the limit
$n=0$. This quantity is said to be duality-invariant if it is well defined in the high temperature
region $\beta<1$ and its temperature dependence is given by a
function $f(\beta)$ which is known analytically, and satifies $f(\beta'=1/\beta)=f(\beta)$.
The simplest example of such quantities is the mean free energy for Gaussian
log-correlated models for which $f(x)=x+ 1/x$. More complicated examples are
presented in \cite{FLDR09}, see e.g. (13)-(14) there. Note that it does not imply anything on
the behaviour of $O_\beta$ for $\beta>1$, hence it is strictly a
property of the high temperature phase. However the property
of duality invariance is not restricted to the replica limit $n=0$.
In \cite{FLDR09} by considering the generating functions for
the moments of the partition function, one notices that
duality invariance can be extended to finite $n$
by further requiring $n' \beta'=n \beta$ (see
(25) there where $s$ should be identified as $- n \beta$).

In the present model, there are more parameters
and we can formulate the duality-invariance property for
the moments as follows:\\

{\it Duality-invariance property:}
\bea \label{duality}
\fl && <y^k>_{\beta,a,b,n} = <y^k>_{\beta',a',b',n'}  \quad \beta'=1/\beta \quad n' = \beta^2 n \quad a'=a/\beta^2
\quad b'= b/\beta^2
\eea
which, again, should be understood in the sense of analytical continuation (i.e.
it does not provide a mapping from the high to low temperature region as discussed above).

\subsubsection{Checking and proving duality for moments.}

The duality property (\ref{duality}) can be checked on the explicit formula (\ref{momentsJ1})-
(\ref{momentsJ2pos}), derived in this paper.  Here we focus on positive moments,
but similar considerations hold for negative ones, when they exist.
To see it explicitly it is convenient to rewrite
that formula in a "geometric" form involving only products over boxes in the
Young diagrams, as
\bea \label{y1}
\fl && ~~~~~~
<y^k>_{\beta,a,b,n} := \left. \left\langle \frac{1}{n} \sum_{r=1}^n y_r^k \right\rangle_{J} ~
\right\vert_{\kappa=-\beta^2} = \left.  \sum_{\lambda, |\lambda|=k}
~~ \lim_{\epsilon \to 0} \frac{k}{n \epsilon} \prod_{s=(i,j) \in \lambda}  B^\epsilon_\lambda(s) \right\vert_{\kappa=-\beta^2}
\eea
where
\bea \label{y2}
\fl  && B^\epsilon_\lambda(s) =
\frac{(j-1 - (i-1) \kappa + \epsilon) (a+ \kappa(n-i)+j) (\kappa(n-i+1)+j-1) }{
(a_{\lambda}(s)+l_{\lambda}(s) \kappa+1)
(a_{\lambda}(s)+l_{\lambda}(s) \kappa+\kappa) (a+b+1+\kappa (2 n - i -1)+j) }  \nonumber  \\
\fl &&
\eea
where $\epsilon$ has been introduced only to remove box (1,1) from one
of the products and the limit $\epsilon \to 0$ is trivial. For application to
the present purpose we need to set $\kappa=-\beta^2$.

The duality is easy to check on that formula and corresponds to the exchange
of the partition $\lambda$ with its dual $\lambda'$. Indeed it is easy to check
that
\bea
\left. \left. B^\epsilon_\lambda(s) \right\vert_{\kappa=-\beta^2} \right\vert_{ \beta, n,a,b} =
\left. \left. B^\epsilon_{\lambda'}(s') \right\vert_{\kappa=-\beta'^2}
\right\vert_{ \beta', n',a',b'}
\eea
where $s'=(j,i)$ is the box in the dual diagram conjugate to $s=(i,j)$,
which implies that arm lengths and leg lengths are also exchanged
under duality. A similar observation over duality-invariant sum over partitions
was reported very recently in \cite{CRS}
for a related problem, about the value at the minimum
of a log-correlated field.

Another proof of the duality invariance property for the moments was provided in the recent work Borodin and Gorin (BG). They proved that these moments are rational functions of their four arguments,
and that the corresponding analytical continuation in these arguments
satisfies an invariance
property, which is equivalent to the
above duality-invariance (\ref{duality}) under the correspondence
(\ref{corr1}).

\subsubsection{Consequence of the duality-invariance: freezing.}

Let us now examine the implications of the relation (\ref{duality}) for the three examples
studied in this paper, showing that a freezing transition at $\beta=1$ is expected in
all cases.

For the GUE problem $a=b=\frac{q+\beta^2}{2}$ and one finds that the duality
invariance in terms of the parameter $q$ can be written as:
\bea
q' = 1 + \frac{q-1}{\beta^2}
\eea
The choice $q=1$ thus ensures duality-invariance of the moments for arbitrary $\beta<1$
and again implies freezing at $\beta=1$.

More generally, starting from Jacobi ensemble measure (\ref{jpdf})
one can ask how to choose $a$ and $b$ so that the model
exhibits the duality-invariance. Consider two (otherwise arbitrary) duality-invariant functions of $\beta$, $\bar a(\beta)$ and $\bar b(\beta)$,
i.e. satisfying $\bar a(1/\beta)=\bar a(\beta)$
and $\bar b(1/\beta)=\bar b(\beta)$, and choose:
\bea \label{aga}
a = \beta \bar a(\beta) \quad , \quad b = \beta \bar b(\beta)
\eea
Then the moments are self-dual. Our second example Eq. (\ref{C}) and (\ref{V0}) of
a log-correlated potential in presence of a background potential, corresponds to
the case of temperature independent constants $\bar a$ and $\bar b$
and its moments are thus duality-invariant.

Finally, for the fBm0 the parameter $a=2 n \beta^2$ and $b=0$. One checks
from (\ref{duality}) that the fBm0 satisfies duality invariance for arbitrary $n$.
In the replica limit $n=0$ we must set $a=b=0$, which is
a self-dual point, hence for this model the moments obey the following
duality-invariance for $\beta<1$
\bea
 <y^k>_{\beta} = <y^k>_{1/\beta}
\eea
so that according to the FDC, one should expect them to exhibit freezing at $\beta=1$.

\subsection{ $n=0$ limit of the moments formula}

The moment formula obtained above, as well as their contour integral representation
are explicit enough to allow for analytic continuation to $n=0$.

\subsubsection{Replica limit of sums over partitions.}

\noindent

\medskip

Consider the formula ({\ref{momentsJ1})-(\ref{momentsJ2}) inserting
$\kappa=-\beta^2$. The limit $n \to 0$ is straightforward, except for one of the factors in the second
line for which we use:
\bea
(- \beta^2 n)_{\lambda_1} = - \beta^2 n (1- \beta^2 n)_{\lambda_1-1} \simeq_{n \to 0}
- \beta^2 n (\lambda_1-1)!
\eea
This leads to the following formula for the disorder averages of the moments
(\ref{momentsab})  of
the general scaled disordered statistical mechanics model
(\ref{modelab})
\be \label{momn0}
\overline{<y^k>_{\beta,a,b}} = \sum_{\lambda, |\lambda|=r} C^{1}_\lambda  C^{2}_\lambda
\ee
with
\bea
\fl &&  C^{1}_\lambda = - \beta^2 k [(\lambda_1-1)! ]^2
\prod_{i=2}^{\ell(\lambda)} [ (\beta^2(i-1))_{\lambda_i} ]^2
\prod_{i=1}^{\ell(\lambda)} \frac{(a+1+ \beta^2 i)_{\lambda_i} }{
  (a+b+2+\beta^2 (i+1))_{\lambda_i} } \\
  \fl && ~~~~~~~~~~~~ \times
\prod_{1 \leq i < j \leq \ell(\lambda)} \frac{\lambda_i-\lambda_j + \beta^2 (i-j)}{\beta^2 (i-j)}
\nonumber  \\
\fl &&  C^2_\lambda
= \prod_{i=1}^{\ell(\lambda)}
 \frac{1}{(1+ \beta^2 (i-\ell(\lambda)))_{\lambda_i} (\beta^2(i-\ell(\lambda)-1))_{\lambda_i} }
   \prod_{1 \leq i < j \leq \ell(\lambda)}
   \frac{(\beta^2(i-j-1))_{\lambda_i-\lambda_j} }{  (1+ \beta^2(i+1-j))_{\lambda_i-\lambda_j}}
    \label{momn1}
\eea
This formula is valid in the higher temperature phase of the model, $\beta<1$, where
all the factors are clearly finite and non-zero. We will study explicitly below a few
moments and their temperature dependence. Note that in the limit $\beta \to 0$
one recovers the moments (\ref{zerokappa}) (i.e. with the weight $P^0(y)$).

The moments of the position of the scaled minimum of the potential
as discussed above are recovered in the zero temperature limit $\beta=+\infty$
of the statistical mechanics model. According to the FDC (see section)
these moments are equal to their value at the freezing transition $\beta=1$.
Hence they can be obtained by taking as limits
\bea \label{limit}
\overline{(y_m)^k} = \lim_{\beta \to 1^-} \overline{<y^k>_{\beta,a,b}}
\eea
of the above expression (\ref{momn0})-(\ref{momn1}).
However in this expression one easily sees that the factor $C^\lambda_2$
has poles for $\beta=1$, while $C^\lambda_1$ is regular and has
a finite limit. Examination
of low moments, detailed below, show massive cancellations of these
poles leading to a well defined finite limit. As shown below, using
the contour integral representation, this limit is indeed finite for any
moment. Note that the poles in the limit $\beta \to 1$ are also present for
$n>0$, so consideration of finite $n$ does not help to handle these
cancellations.

The formula (\ref{limit}) together with formula (\ref{momn0})-(\ref{momn1})
thus gives arbitrary positive integer moments of the position of the
global minimum of the log-correlated process
and as such is a main result of our paper. They can be used
to generate these moments to a very high degree on the computer.

\subsubsection{Contour integral representation of moments for $n=0$.}

\noindent

\medskip

(i) {\it positive moments.}
To take the limit $n=0$ in the contour integral
formula (\ref{contour1}) we rewrite
\bea
\fl \frac{1}{n \beta^2} \prod_{i=1}^k \frac{u_i+\beta^2}{u_i + \beta^2 (1-n)}  = \frac{1}{n \beta^2} \prod_{i=1}^k
(1 +  \frac{n \beta^2}{u_i + \beta^2 (1-n)} )
= \frac{1}{n \beta^2} +  \sum_{i=1}^k \frac{1}{u_i + \beta^2} + O(n) \nonumber
\eea
Inserting the first term in the contour integral gives zero. Next, it is easy to see that only the
pole in $u_1$ gives non zero residue. Hence we can insert only this term in the
integral (\ref{contour1}), where we can now safely take $n=0$ leading to
the following representation for the disorder average:
\bea \label{contourbeta}
\fl && \overline{<y^k>_{\beta,a,b}} =
\int \prod_{i=1}^k \frac{du_i}{2 i \pi}
 \prod_{1 \leq i < j \leq k} \frac{u_j-u_i}{(u_j-u_i + \beta^2)(u_j-u_i +1)} \\
\fl  && \times
\prod_{1 \leq i +1 < j \leq k} (u_j-u_i + 1 + \beta^2)  \frac{1}{u_1} \times \prod_{i=1}^k
\frac{u_i - 1- a - \beta^2 }{u_i - 2 - a - b - 2 \beta^2}
\eea
where we have shifted $u_i \to u_i - \beta^2$ for convenience. This again must be supplemented with the condition (i) $C_1:$
$|u_1| \ll |u_2| \ll |u_3| ..\ll |u_p|$ and (ii) $C_2$: all the contours enclose the singularities at
$u_1 = 0$ but not at $u_i = 2 +a+b+2 \beta^2$. In practice the contours will
run (and close) in the negative half plane $Re(u_i)<0$ and pick up residues from poles on
the negative real line. We note that one needs the condition
$ 2 +a+b+2 \beta^2 >0$ which, for $a=b$ is precisely the one found in
\cite{FLDR09} corresponding to a binding transition to the edge
(for $a< -1 - \beta^2$). We will thus assume that the condition is
fulfilled, which is the case for all three examples considered
here.

The limit $\beta \to 1$ can be performed explicitly leading
to a contour integral representation for the positive integer moments of
the position of the global extremum of the log-correlated field
\bea \label{contour0}
\fl && \overline{ y_m^k } =
\int \prod_{i=1}^k \frac{du_i}{2 i \pi}
 \prod_{1 \leq i < j \leq k} \frac{u_j-u_i}{(u_j-u_i +1)^2}
\prod_{1 \leq i +1 < j \leq k} (u_j-u_i + 2)  \frac{1}{u_1}  \prod_{i=1}^k
\frac{u_i - 2- a  }{u_i - 4 - a - b } \nonumber \\
\fl &&
\eea
with the same contour conditions $C_1$ and $C_2$.
This formula should thus be equivalent to our main result
(\ref{limit})-(\ref{momn0})-(\ref{momn1}), which we
have checked for a few low order moments.

(i) {\it negative moments.} The same manipulation as
above in the limit $n \to 0$ gives the disorder averaged
moments for $k \geq 1$:
\bea
\fl && \overline{< y^{-k}>_{\beta,a,b}}
=
\int \prod_{i=1}^k \frac{du_i}{2 i \pi}
 \prod_{1 \leq i < j \leq k} \frac{u_j-u_i}{(u_j-u_i - \beta^2)(u_j-u_i -1)} \nonumber  \\
\fl  && \times
\prod_{1 \leq i +1 < j \leq k} (u_j-u_i - 1 - \beta^2) \frac{1}{u_1}
\times \prod_{i=1}^k \frac{u_i -1 - a - b - \beta^2}{u_i-a} \nonumber   \\
\fl && \label{contourneg1}
\eea
provided these moment exist.
Taking again the limit $\beta \to 1$ one obtains the negative integer moments of
the position of the global extremum of the log-correlated field as
\bea
\fl && \overline{y_m^{-k}}
=
\int \prod_{i=1}^k \frac{du_i}{2 i \pi}
 \prod_{1 \leq i < j \leq k} \frac{u_j-u_i}{(u_j-u_i -1)^2} \prod_{1 \leq i +1 < j \leq k} (u_j-u_i - 2) \frac{1}{u_1}
\times \prod_{i=1}^k \frac{u_i -2 - a - b }{u_i-a} \nonumber   \\
\fl && \label{contourneg2}
\eea
for $k$ positive integer, provided they exist. In both integrals the
contours obey the two conditions (i) $C_1$:  $|u_1| \ll |u_2| \ll |u_3| ..\ll |u_k|$ and (ii) $C_2$:
all the contours enclose the singularities at $0$ and not at $a$.
At present there is no equivalent formula in terms of sums over partitions,
hence the above formula is an important result of the paper.

We now turn to explicit study of the low moments

\subsection{Calculation and results for the first moment.}

\noindent

\medskip

Let us illustrate the calculation using the contour integral (\ref{contour1}) on the simplest example of the
first moment $k=1$
\bea
\fl ~~  < y>_{\beta,a,b,n}  &=& \frac{1}{n \beta^2}
\int  \frac{du_1}{2 i \pi}  \frac{u_1+\beta^2}{u_1 + \beta^2 (1-n)}
\frac{u_1 - 1- a }{u_1 - 2 - a - b - \beta^2(1-n)} \nonumber  \\
\fl ~~ & = & \frac{1+ a- \beta ^2 (n-1)}{2+ a+b-2 \beta ^2 (n-1)} \label{res11}
\eea
which is equal to the residue at $u_1=- (1-n) \beta^2$. In terms of
partitions only the partition $\lambda=(1)$ contributes, so it is
easy to see on (\ref{momentsJ1})-(\ref{momentsJ2}), and even
more immediate on (\ref{y1})-(\ref{y2}) [using $i=j=1$, $a_\lambda=\ell_\lambda=0$],
that it reproduces (\ref{res1}).
One can explicitly verify on this result that the first moment is invariant by the duality transformation
(\ref{duality}).

The $n=0$ limit yields the disorder averaged first moment
\bea
 \overline{<y>_{\beta,a,b}} = \frac{1+ a+ \beta ^2}{2+ a+b+2 \beta ^2}
\eea
For the symmetric situation $a=b$, which is the case both for the
 fBm $a=b=0$ and for the GUE characteristic polynomial $a=b=\frac{1 + \beta^2}{2}$
 the first moment is thus simply
\bea
 \overline{<y>_{\beta}}=\frac{1}{2}
\eea
In the second example of the LCGP with edge charges one obtains:
 \bea
 \overline{<y>_{\beta}}= \frac{1+ \bar a \beta + \beta ^2}{2+ (\bar a+\bar b) \beta +2 \beta ^2}
\eea
The duality-freezing conjecture then leads to the the first moment of the
position of the minimum
 \bea \label{mom1aa}
 \overline{y_m} = \frac{2+ \bar a}{4+ \bar a+\bar b} \quad , \quad \overline{y_m} - \frac{1}{2} =
\frac{\bar a - \bar b}{2 (\bar a + \bar b + 4)}
\eea
which is Eq. (\ref{p3-1}) in {\bf Prediction 2}.

These results can be compared to the first moment in absence of the random potential and
at finite inverse temperature $\beta$, i.e. from the measure $P^0(y)|_{a=\beta \bar a, b=\beta \bar b}$
\bea
<y >_{P^0} - \frac{1}{2} = \frac{\beta( \bar a- \bar b)}{2 (2+\beta( \bar a+\bar b) )}
\eea
which reproduces the absolute minimum $y_m^0=\bar a/(\bar a+\bar b)$ in the absence of disorder for
$\beta=+\infty$. Comparing with (\ref{mom1aa}) shows that
even at the freezing temperature $\beta=1$, disorder brings the
average position closer to the midpoint $y=\frac{1}{2}$.

\subsection{Results for second, third and fourth moments}

\noindent

\medskip

The calculation of the second moment by the contour integral method
is relatively simple, and sketched in the Appendix B section \ref{app:second} for $n=0$. It leads to
the disorder average:
\bea \label{y22}
\fl \overline{<y^2>_{\beta,a,b}} = \frac{\left(a+\beta ^2+1\right) \left(\beta ^2 (4 a+2
   b+9)+(a+2) (a+b+2)+4 \beta ^4\right)}{\left(a+2
   \beta ^2+b+2\right) \left(a+2 \beta ^2+b+3\right)
   \left(a+3 \beta ^2+b+2\right)}
   \eea
This expression is also easily recovered from the sum (\ref{momentsJ1})-(\ref{momentsJ2}), or
 (\ref{y1})-(\ref{y2}) involving now the two partitions $(2)$ and $(1,1)$.  For $\beta=0$ it
 reproduces (\ref{zerokappa}), i.e. the trivial average with
 respect to the weight $P^0(y) \sim y^a(1-y)^b$. The expression (\ref{y22})
 is duality-invariant and we expect that it freezes at $\beta=1$ leading to the predictions for
the second moment of the position of the extremum.

Let us now detail the results for each example separately, including moments
up to $k=4$ when space permits, more detailed derivations and
results being displayed in the Appendix B.

\subsubsection{Log-correlated potential with edge charges $\bar a,\bar b$.}

The second moment of the position of the global minimum
is obtained from above as \footnote{One sees on this expression that binding to the edge $y_m \to 0^+$ (resp. $y_m \to 1^-$)
occurs when $a \to -2^+$ (resp. $\bar b \to -2$) so one may surmise that the result is
valid as long as $-2<\bar a,\bar b$.}
\bea
&& \overline{y_m^2} = \frac{(\bar a+2) (\bar a (\bar a+\bar b+8)+4 \bar b+17)}{(\bar a+\bar b+4) (\bar a+\bar b+5)^2}
\eea
This leads to the variance (\ref{p3-2}) in {\bf Prediction 2}, replacing there $x \to y$.
Expressions for higher moments for general $\bar a,\bar b$ are too bulky to present here, and we
only display the third moment in (\ref{3cum}) and the skewness
in (\ref{Skew1}).

Two special cases are of interest:

(i) {\it only one edge charge, at $y=0$:} One sets $\bar b=0$. Let us give here the skewness
in that case
\bea
&& Sk := \frac{\overline{(y_m - \overline{y_m})^3}}{\overline{(y_m - \overline{y_m})^2}^{~\frac{3}{2}}}
= -\frac{\bar a (\bar a+5) (\bar a (7 \bar a+68)+164)}{\sqrt{2} \sqrt{\bar a+2} (\bar a+6)^2 (2 \bar a+9)^{3/2}}
\eea
which is negative. It can be compared with the skewness associated to
the measure $\sim y^{\bar a}$, which is
\bea
&& Sk_0 = -\frac{2 \bar a \sqrt{\bar a+3}}{\sqrt{\bar a+1} (\bar a+4)}
\eea
One finds that $Sk$ decreases from $0$ to $-7/4$ as $\bar a$ increases,
while $Sk_0$ decreases from $0$ to $-2$, hence they are quite distinct.

(ii) {\it symmetric case $\bar a=\bar b$.} The variance is obtained as:
\bea
&& \overline{y_m^2} - \overline{y_m}^2 = \frac{4 \bar a+9}{4 (2 \bar a+5)^2}
\eea
where we recall, $\overline{y_m}=\frac{1}{2}$, and
we checked explicity that the moment formulae lead to
\bea \label{centered}
\overline{ \bigg( y_m - \frac{1}{2} \bigg)^3 }  = 0
\eea
as expected by symmetry $y \to 1-y$. The fourth moment
and kurtosis are obtained as
\bea
&& \overline{ y_m^4 }  = \frac{4 a^5+84 a^4+663 a^3+2488 a^2+4478 a+3110}{4 (a+3) (2
   a+5)^2 (2 a+7)^2} \label{fourth} \\
&& {\rm Ku} = -\frac{2 \left(8 a^5+248 a^4+2054 a^3+7328 a^2+12053 a+7523\right)}{(a+3) (2 a+7)^2 (4
   a+9)^2}
\eea
The kurtosis can be compared to the one of the measure $y^{\bar a} (1-y)^{\bar a}$ which is:
\bea
&& \kappa_0 = - \frac{6}{5+2 a}
\eea
Note that when $\bar a$ increases, ${\rm Ku} \to -1/4$ while ${\rm Ku}_0 \to 0$.

\subsubsection{GUE characteristic polynomial.}
For the GUE-CP we must insert $a=b=\frac{1+\beta^2}{2}$,
in (\ref{y2}). The second
moment for the associated statistical mechanics model in the high
temperature phase $\beta <1$, and for the position of the global
minimum (obtained by setting $\beta=1$) are then found to be:
\bea
&& \overline{<y^2>_\beta}  = \frac{15 \beta ^4+32 \beta ^2+15}{4 \left(3 \beta
   ^2+4\right) \left(4 \beta ^2+3\right)} \quad , \quad \overline{y_m^2}  = \frac{31}{98}
\eea
For completeness the expression at finite $n$ is given in (\ref{GUEn}).
In the original variable $x \in [-1,1]$, i.e. the support of the semi-circle,
using $x=1-2 y$ and the result for the first moment we obtain:
\bea
&& \overline{<x^2>_\beta} = \frac{3 + 7 \beta^2 + 3 \beta^4}{(4 + 3 \beta^2) (3 + 4 \beta^2)}
\quad , \quad \overline{x_m^2}  =  \frac{13}{49} = 0.265306.. \label{var}
\eea
where we expect now the PDF of the position of the maximum, ${\cal P}(x)$,
to be centered and symmetric
around $x=0$ (which was checked explicitly up to fifth moment).

We give directly the fourth moment of the position of the maximum
(see \ref{app:4} for finite temperature expressions)
\bea
&& \overline{y_m^4} = \frac{401}{2352} \quad , \quad  \overline{x_m^4} =  \frac{20}{147}  = 0.136054..
\eea

which leads to the fourth cumulant and kurtosis as:

\be
\overline{x_m^4}^c =  - \frac{541}{7203} \quad , \quad
{\rm Ku} = - \frac{541}{507} = -1.06706..
\ee

These moments can be compared with the ones of the semi-circle density
\bea
<x^k>_\rho := \int_{{ -1}}^1 dx x^k \rho(x) \quad , \quad \rho(x) = \frac{2}{\pi} \sqrt{1-x^2}
\eea
which are:
\bea
&& <x^2>_\rho = \frac{1}{4} \quad, \quad  <x^4>_\rho = \frac{1}{8} = 0.125 \quad , \quad  {\rm Ku}  = -1
\eea
and are found quite close, suggesting that ${\cal P}(x)$ is distinct from, but numerically
close, to the semi-circle density.

\subsubsection{Fractional Brownian motion.} For the fBm0 we should set $a=b=0$
in (\ref{y22}) leading to the following disorder average of the associated
statistical mechanics model in the high temperature phase:
\bea
\fl &&  \overline{<y^2>_\beta} = \frac{4 \beta ^4+9 \beta ^2+4}{2 \left(2 \beta
   ^2+3\right) \left(3 \beta ^2+2\right)} \quad , \quad
 \overline{<y^2>_\beta} - \overline{<y>_\beta}^2 =
\frac{\left(\beta ^2+2\right) \left(2 \beta ^2+1\right)}{4 \left(2 \beta ^2+3\right)
   \left(3 \beta ^2+2\right)}
\eea
which is manifestly duality-invariant (see (\ref{fBmn}) for the $n$-dependence).  The
second moment of the position of the global minimum of the fBm0
is thus predicted to be
\bea
&&  \overline{y_m^2} = \frac{17}{50} \quad , \quad
\overline{y_m^2} - \overline{y_m}^2 = \frac{9}{100}
\eea
as displayed in {\bf Prediction 3}. This is
distinct, but numerically not very different, from what is
obtained from a uniform distribution $P^0(y)=1$ on $[0,1]$, namely
$<y^2>_{P^0}=\frac{1}{3}$ and $<(y-\frac{1}{2})^2>_{P^0}=\frac{1}{12}$.

All odd moments centered around $y=\frac{1}{2}$ are predicted
to vanish, as in (\ref{centered}). Although we explicitly checked
up to fifth it should be a general property. The asymmetry
induced by fixing one point of the fBm0 at $x=0$ and letting the
one at $x=1$ free, does not manifest itself in the moments
of the minimum (or at any temperature in the statistical mechanics
model). It does arise however to first order in $n$ (as seen e.g.
from (\ref{res1})) and is detectable in the joint distribution of values and
positions of the minimum (see below).

From the formula (\ref{fourth}) specialized to $a=0$ we obtain
respectively the fourth moment, the fourth cumulant
and the kurtosis $\kappa$ for the position of the
minimum of the fBm0:
\bea
\fl && ~~~~~ \overline{y_m^4}  = \frac{311}{1470} 
\quad , \quad  \overline{y_m^4}^c  := \overline{ \bigg( y_m - \frac{1}{2} \bigg)^4 } - 3 \overline{ \bigg( y_m - \frac{1}{2} \bigg)^2 }^2
= \frac{- 7523}{735000} 
\\
\fl && ~~~~~~~~~~~~~~~~~~~~~~~~~~~~~\kappa : =\frac{\overline{y_m^4}^c}{\overline{y_m^2}^2} = \frac{- 15046}{11907} = -1.26363..
\eea
These three numbers are respectively $\frac{1}{5}$, $- \frac{1}{120}$ and $-1.2$ for
a uniform distribution of $[0.1]$, hence a difference of a few percent.\\

\subsection{Results for negative moments}

The negative moments
lead interesting additional information on the three problems
under study.
Let us present the (short) calculation of the first negative moment,
and also give the expression for the second.
Equation (\ref{contourneg}) for $k=1$ yields
\bea
\fl  < y^{-1}>_{\beta,a,b,n}  &= & \frac{-1}{n \beta^2}
\int \frac{du_1}{2 i \pi}
 \frac{u_1-n \beta^2}{u_1 } \times
\frac{u_1 -1 - a - b - \beta^2(1-n)}{u_1-a} \\
\fl &= &\frac{1+a+b+ \beta^2(1-n)}{a}
\eea
The continuation to $n=0$ of this formula, and of the one for the second negative moment
(calculated in Appendix B, section \ref{app:2neg} )
leads to the following disorder averages
in the statistical mechanics model for $\beta^2 < 1$
\bea \label{inverse12}
\fl && \overline{<y^{-1}>_\beta} = \frac{1+a+b+ \beta^2}{a} \quad , \quad
\overline{< y^{-2}>_\beta}  =
\frac{\left(a+\beta ^2+b+1\right) \left(a (a+b)+\beta^2\right)}{(a-1) a \left(a-\beta ^2\right)}
\eea
Obviously the same formula exist hold exchanging $y \to 1-y$ and $a \to b$.
 One can check that these formula coincide with the ones obtained from
the general result (\ref{momentsJ2}-\ref{momentsJ2neg}).

One can check that for $\beta=0$ and $a,b$ fixed, i.e. in the absence of the random potential, these
formula agree with the same averages over the
deterministic measure $P^0(y) \sim y^a (1-y)^b$,
i.e. Eq. (\ref{zerokappa}) setting $k=-1,-2$ there.
The effect of the disorder is thus to {\it increase} the values
of the inverse moments,
presumably from the events when favorable regions
in the random potential appear near the edges.
We see that $a>0$ (repulsive charge at $y=0$) is required for the finiteness
of the first inverse moment, and $a>1$ for the finiteness of the second.
In addition, since $y \in [0,1]$, one must have
$\overline{<y^{-1}>_\beta} \geq 1$. Hence a binding transition at $y=1$ must occur
when the charge becomes too attractive, for  $b \leq b_c = - 1 - \beta^2$.
Symmetric conditions hold under exchanges of $y \to 1-y$ and $(a,b) \to (b,a)$.

These moments give some information about the disorder-averaged
Gibbs measure: if we assume a power law behavior near the edge,
$\overline{p_\beta(y)} \sim_{y \to 0} y^c$, the exponent $c=c(a,b,\beta)$
must be such that $c>0$ whenever $a>0$, and $c>1$ whenever $a>1$.
From the divergence of the inverse moments when $a$ reaches these
values, we can surmise that $c$ also vanishes at $a=0$ and equals $1$ at $a=1$.
The simplest possible scenario then
is that $c=a$, but it remains to be confirmed.

The effect of disorder saturates at $\beta=1$ where we predict freezing
in these negative moments, which, as can be checked on (\ref{inverse12})
using (\ref{duality}) are duality-invariant. As discussed above we
predict that the full PDF $\overline{p_\beta(y)}$ freezes, i.e.
${\cal P}(y_m) = \overline{p_{\beta=1}(y)}$.
Let us now discuss consequences for our three examples.

(i) {\it For the GUE-CP}, one must set $a=b=\frac{1+ \beta^2}{2}$. One sees that
the first inverse moment exist, but not the second, so the exponent $0 <c <1$.
One finds that the first inverse moment is temperature independent:
\bea
&& \overline{<y^{-1}>_\beta} = 4  \quad , \quad 0 \leq \beta \leq 1
\eea
which leads to the following predictions for the position of the maximum
\bea \label{predinv}
&& \overline{y_m^{-1}} = 4 \quad , \quad \overline{(1-x_m)^{-1}} = 2
\eea
Remarkably, this is also exactly the value of the same average with
respect to the semi-circle density
\bea
<\frac{1}{1-x}>_\rho = 2
\eea
This suggests that the two distributions, ${\cal P}(x)$ ands $\rho(x)$, although distinct, are very
similar near the edges.

(ii) {\it LCP with edge charges}, one must set $a=\beta \bar a$, $b= \beta \bar b$.
The prediction for the first two inverse moments of the position of the global minimum
is:
\bea
\overline{y_m^{-1}} = \frac{2+\bar a+\bar b}{\bar a} \quad , \quad  \overline{y_m^{-2}} =
\frac{\left(\bar a+\bar b+2\right) \left(\bar a (\bar a+\bar b)+1\right)}{\bar a (\bar a-1)^2}
\eea
where domain of existence has been discussed above. It would be interesting
to see whether the simplest scenario for the edge behavior, i.e. that ${\cal P}(y_m) \sim y_m^a$ near $y=0$,
and by symmetry ${\cal P}(y_m) \sim y_m^b$ near $y=1$ can be
confirmed (or infirmed) in future studies.

(iii) {\it For the fBm0} one must set $a=b=0$ (for $n=0$), and neither of these moments exist. Hence the
edge exponent of ${\cal P}(y) \sim y^c$ is such that $c \leq 0$ (and probably $c=0$).

\subsection{Correlation between position and value of the minimum}

Until now we used only the values of the moments at $n=0$. However they do exhibit a non-trivial dependence
in the number of replica $n$. One may thus ask what is the information encoded in that dependence.

The detailed analysis is performed in Appendix H.
The answer is that the knowledge of the
$n$-dependence of all moments
allows in principle to reconstruct the joint distribution, ${\cal P}(x_m,V_m)$,
of the position and value of the extremum. This is an ambitious task which is
far from completed. However in Appendix H
we give a general formula
for the {\it conditional moments}, i.e the moments of $x_m$ conditioned
to a particular value of $V_m$.

Here we display the results for the simplest
cross-correlations in the LCGP model with edge charges,
the derivation and more results are given in Appendix H.
We find
\bea
&& \overline{y_m V_m} - \overline{y_m} ~~ \overline{ V_m}  =
\frac{\bar b- \bar a}{(\bar a+\bar b+ 4)^2} \\
&&
\overline{ (y_m - \overline{ y_m} ) (V_m - \overline{V_m})^2 } =
\frac{4 (\bar a-\bar b)}{(\bar a+\bar b+4)^3}
\eea

In the case $\bar a=\bar b$ these two correlations of the first moment
vanish, and so do higher ones: one shows (from the general formula in Appendix H)
 that the first conditional moment, $\mathbb{E}(y_m | V_m)= \frac{1}{2}$
independently of $V_m$. Continuing with the case $\bar a=\bar b$
we further obtain
\bea
&& \overline{(y_m - \frac{1}{2}) V_m}  =  0 \\
&& \overline{(y_m^2 - \overline{y_m^2}) (V_m-\overline{ V_m} )} =
\frac{(\bar a+2) (2 \bar a+1)}{2 (2 \bar a+5)^3} \\
&& \overline{ (y_m^2 - \overline{ y_m^2} ) (V_m - \overline{V_m})^2 } = -\frac{4 \bar a^2+8 \bar a+1}{(2 \bar a+5)^4}
\eea
Setting $\bar a=1$ leads to the prediction for the
GUE-CP as
\bea
&& \overline{x_m (V_m- \overline{V_m})}   =   0  \quad , \quad \overline{x_m (V_m- \overline{V_m})^2}   =   0 \\
&& \overline{(y_m^2 - \overline{y_m^2}) (V_m-\overline{ V_m} )} =
 \frac{9}{686}
\quad , \quad
\overline{(x_m^2 - \overline{x_m^2}) (V_m-\overline{ V_m} )}
 =  \frac{18}{343} \\
&& \overline{(x_m^2 - \overline{x_m^2}) (V_m-\overline{ V_m} )^2}
 =  - \frac{52}{2401}
\eea
and we recall that $x_m=1-2 y_m$ with $\overline{y_m}=\frac{1}{2}$ for the GUE-CP
and in fact, as discussed above the first conditional moment
$\mathbb{E}(x_m| V_m)= 0$ vanishes
for any $V_m$.

The fBm0, as we defined it with the value fixed at $y=0$,
provides an interesting example of a process
with non-trivial correlation. Indeed the above formula must be modified
since $a =2 \beta^2 n$ for the fBm0.  As discussed in
Appendix G, that leads to difficulties in the method for
the determining the PDF of $V_m$. One should thus be careful in
assessing the validity of the following results
for the case of the fBm0. They read
\bea
&& \overline{ (y_m - \overline{ y_m} ) (V_m - \overline{V_m}) }  = - \frac{1}{4}  \\
&& \overline{ (y_m - \overline{ y_m} ) (V_m - \overline{V_m})^2 }  = 0  \\
&& \overline{ (y_m^2 - \overline{ y_m^2} ) (V_m - \overline{V_m}) } = - \frac{21}{100} \\
&& \overline{ (y_m^2 - \overline{ y_m^2} ) (V_m - \overline{V_m})^2 } = \frac{2}{25}
\eea
and we recall $\overline{y_m}=\frac{1}{2}$ for the fBm.
The negative value obtained for the first correlation (first line) is a reflection of the boundary condition
chosen, namely pinning of $B_0^{(\eta)}(y=0)=0$ and free
boundary condition at $y=1$, which allows for lower values
of the minimum near the right edge. The vanishing of the
second line follows from the discussion in Appendix H.

\section{Other ensembles}

\subsection{General considerations}

Once the moments for the Jacobi ensemble are known, one can obtain moments
in a few other ensembles. Define
the generic measure
\be
  {\cal P}_A({\bf y})  d {\bf y} = \frac{1}{{\cal Z}^A_n} \prod_{i=1}^n dy_i \mu_A (y_i)
\prod_{1 \leq i < j \leq n} |y_i-y_j|^{2 \kappa}
\ee
with, for Jacobi, $\mu_A(y) = \mu_J(y) := y^a (1-y)^b \theta(0<y<1)$. The main other ensembles
differ only by the choice of the weight function $\mu_A(y)$.

\begin{enumerate}

\item

{\it Laguerre ensemble} Define $y_i=z_i/b$ and take the limit $b \to +\infty$.
Then
\bea
\lim_{b \to +\infty} {\cal P}_J({\bf y})  d {\bf y} = {\cal P}_L({\bf z})  d {\bf z} \quad , \quad
\mu_A(z) = z^a e^{-z} \theta(z)
\eea
where ${\cal Z}^L_n= \lim_{b \to +\infty} b^{a+n + \kappa n(n-1)/2} {\cal Z}^J_n
= \prod_{j=1}^n \frac{\Gamma(1+a+ (j-1) \kappa) \Gamma(1+j \kappa)}{\Gamma(1+\kappa)}$.

Hence the moments in Laguerre ensemble are obtained as:
\bea
<z^k>_L = \lim_{b \to \infty} b^k <y^k>_J   \label{lag1}
\eea

However in our statistical mechanics model, for instance
the LCG random potential with an external background, we need to
define $z$ slightly differently, i.e. $z=\frac{b}{\beta} y = \bar b y$. This ensures
duality invariance of the problem, and the Laguerre weight can
now be interpreted as a bona-fide external background potential:
\bea \label{VL}
V_0(z) = - \bar a \ln z + z
\eea
which confines the particle near the edge $z=0$ (for $\bar a>0$).

\item

{\it Gaussian-Hermite ensemble} Define $y_i=\frac{1}{2} + (z_i/\sqrt{8 a})$ and take the limit $a=b \to +\infty$.
Then
\bea
\lim_{a=b \to +\infty} {\cal P}_J({\bf y})  d {\bf y} = {\cal P}_G({\bf z})  d {\bf z} \quad , \quad
\mu_G(z) = e^{-z^2/2}
\eea
where ${\cal Z}^G_n= \lim_{a=b \to +\infty} 4^{an} (8a)^{n/2  + \kappa n(n-1)/2} {\cal Z}^J_n
= (2 \pi)^{n/2} \prod_{j=1}^n \frac{\Gamma(1+ j \kappa)}{\Gamma(1+\kappa)}$ is the Mehta integral.
Similarly below we will introduce a factor of $\beta$ in the definition, see
next section.

\item

{\it Inverse-Jacobi weights} Define $y_i=1/z_i$
Then
\bea
{\cal P}_J({\bf y})  d {\bf y} = {\cal P}_L({\bf z})  d {\bf z} \quad , \quad
\mu_A(z) = z^{-2-a-b-2(n-1) \kappa} (z-1)^b \theta(z-1)  \nonumber
\eea
where ${\cal Z}^I_n= {\cal Z}^J_n$

\end{enumerate}

\noindent

\medskip

There is a second set of models, for which the correspondence is less direct.
We will follow the arguments of \cite{ForWar} to surmise a relation between
moments. As in that work, one starts
with the simple identity, for $k \in \mathbb{Z}$, $a \in \mathbb{R}$
\bea
\fl && \int_{-\pi}^\pi d\theta e^{i \theta (a +1 +k) } = \frac{2 \sin((1+a+k) \pi)}{1+a+k} = 2 (-1)^k \sin((1+a) \pi)
\int_0^1 dt ~ t^{a+k}
\eea
valid whenever the last integral converge. That leads to the multiple-integral version
\bea
\fl && \prod_{j=1}^n \int_{-\pi}^\pi d\theta_j e^{i \theta_j (a +1) } f(-e^{i \theta_1},..-e^{i \theta_n})
= [2 \sin((1+a) \pi)]^n \prod_{j=1}^n \int_{0}^1 dt_j t_j^a f(t_1,..t_n) \nonumber
\eea
for any Laurent polynomial $f$. From this one conjectures interesting relations between
quantities in the circular and Jacobi ensembles, see (1.15)-(1.17) in \cite{ForWar}
as well as Proposition 13.1.4 in Chap.13 of \cite{Forbook}.
Further elaborations of these relations lead us to the following conjectures
for the moments:

\begin{enumerate}

\item {\it Circular ensemble with weight.}
Consider the CUE with weight, defined by the joint probability
\be
 \frac{1}{{\cal Z}^C_n} \prod_{i=1}^n \frac{d\theta_i}{2 \pi} |1 + e^{i \theta_i}|^{2 \mu}
\prod_{1 \leq i < j \leq n} |e^{i \theta_i} - e^{i \theta_j}|^{2 \kappa}
\ee
for the variables $\theta_i \in [-\pi,\pi[$.
Then we conjecture that
\bea
< \cos(k \theta) >_{\text{circular}} = (-1)^k < y^k >_{\kappa,a,b,n} |_{a=-\mu-1-\kappa(n-1),b =2 \mu}
\eea

\item {\it Cauchy-$\beta$ ensemble.}
Following (1.19) in \cite{ForWar} and using
the stereographic projection from the circle to the real axis, $e^{i \theta}=(i - z)/(i+z)$
we obtain the Cauchy-$\beta$ ensemble which has weight on the whole real axis $z \in \mathbb{R}$
\bea
\mu_C(z) = \frac{1}{(1+z^2)^\rho}  \quad , \quad \rho = 1 + \mu + (n-1) \kappa
\eea
For this ensemble, the conjecture then becomes
\bea
<  Re[ (\frac{i - z}{i+z})^k ] >_{\text{Cauchy}} = (-1)^k < y^k >_{\kappa,a,b,n} |_{a=-\rho,b =2 \rho-2 - 2 (n-1) \kappa}
\eea
which we checked is obeyed for $\kappa=0$, in which case
one has $< y^k >_{J,\kappa=0}= (1-\rho)_k/(\rho)_k$.
Note that interesting integrable generalization of Cauchy ensemble
was proposed in \cite{BO}.

\end{enumerate}

Let us now study the two following examples in more details

\subsection{Moments for the Laguerre ensemble}

\subsubsection{General formula.}

 From (\ref{momentsJ1})-(\ref{momentsJ2pos}), performing the limit (\ref{lag1}) we obtain the general
positive moments of the Laguerre ensemble as\footnote{Note that positive moments for the Laguerre ensemble were also presented in
\cite{MR15} in an equivalent, but less explicit form.}
\be \label{momentsJ1L}
 \left\langle \frac{1}{n} \sum_{j=1}^n y_j^k \right\rangle_L = \sum_{\lambda, |\lambda|=k} A_\lambda a^+_\lambda
\ee
where the sum is over all partitions of $k$, and
\bea \label{momentsJ2L}
\fl && A_\lambda a^+_\lambda = \frac{ k (\lambda_1-1)! }{
(\kappa(\ell(\lambda)-1)+1)_{\lambda_1}}
\prod_{i=2}^{\ell(\lambda)}
 \frac{(\kappa(1-i))_{\lambda_i}}{(\kappa(\ell(\lambda)-i)+1)_{\lambda_i}}  \prod_{1 \leq i < j \leq \ell(\lambda)}
\frac{\kappa(j-i)+\lambda_i-\lambda_j}{\kappa(j-i) }   \\
\fl && \times \frac{1}{n} \prod_{i=1}^{\ell(\lambda)} (a+1+ \kappa(n-i))_{\lambda_i}
 \frac{(\kappa(n-i+1))_{\lambda_i} }{ (\kappa(\ell(\lambda)-i+1))_{\lambda_i} }
   \prod_{1 \leq i < j \leq \ell(\lambda)} \frac{(\kappa(j-i+1))_{\lambda_i-\lambda_j} }{ (\kappa(j-i-1)+1)_{\lambda_i-\lambda_j}  }
   \nonumber
\eea
i.e. a single factor has disappeared. It turns out that these sums are polynomials, although it
may not be easy to see on this expression.

Let us give the first two moments:
\bea
&& <z>_L = 1 + a + \kappa (n-1) \\
&& <z^2>_L = (1 + a + \kappa (n-1))(2 + a + 2 \kappa (n-1))
\eea
higher moments become more complicated polynomials.  Formula for
negative moments can also be obtained from (\ref{momentsJ1})-(\ref{momentsJ2neg}) performing the
same limit.

\subsubsection{Random statistical mechanics model associated to Laguerre ensemble.}

In the corresponding disordered model, LCRG with a background confining
potential (\ref{VL}) we find
\bea
&& \overline{<z>_\beta} = \frac{1}{\beta} + \bar a + \beta  \\
&& \overline{<z^2>_\beta} - \overline{<z>}^2 = (\frac{1}{\beta} + \beta) (\frac{1}{\beta} + \beta + \bar a)
\eea
Freezing of these manifestly duality invariant expressions lead to
\bea
&& \overline{z_m} = 2 + \bar a  \quad , \quad
\overline{z_m^2} - \overline{z_m}^2 = 2 (2 + \bar a)
\eea
In the absence of disorder the absolute minimum is at
$z_m^0=\bar a$, hence the random potential now tends to push the minimum towards
the larger positive $z$ (i.e. to unbind the particle).
One finds a few higher moments
\bea
\fl && \overline{z_m^3} = (2 + a) (23 + a (10 + a)) \quad , \quad
\overline{z_m^4} = (2 + a) (168 + a (99 + a (18 + a)))
\eea
whose associated cumulants have simpler expressions
\bea
\fl && \overline{z_m^3}^c = 7 (2 + \bar a) \quad , \quad
\overline{z_m^4}^c = (\bar a-32)(2+\bar a) \\
\fl &&
\overline{z_m^5}^c =  4 (2 + \bar a) (42 - 5 \bar a)  \quad , \quad
\overline{z_m^6}^c =  2 (2 + \bar a) (458 + \bar a (-147 + 2 \bar a))
\eea

\subsection{Gaussian-Hermite ensemble}

\label{sec:gauss}

Let us now turn to the Gaussian ensemble.
This model is of great interest as its disordered statistical mechanics
is associated to the solution of the one-dimensional decaying Burgers equation
with a random initial condition and of viscosity $\sim 1/\beta$.
The velocity is the gradient of a potential, and the initial condition is chosen
to correspond to a log-correlated random potential.
The one-space point statistics of the velocity at any later time is then exactly
associated to the statistical model of the Gaussian ensemble, as
we showed in \cite{FLDR10}. The freezing
transition at $\beta=1$ corresponds to a transition in the Burgers dynamics
from a Gaussian phase, to a shock-dominated phase.
For more details on the correspondence between the two problems, we refer the reader to \cite{FLDR10}
where the model is introduced and analyzed. We use the same
conventions as in that work. Defining the new variable $z$:
\bea  \label{zy}
z = \sqrt{8 \bar a} (y-\frac{1}{2})   \quad , \quad a = \beta \bar a
\eea
in terms of the Jacobi variable $y \in [0,1]$, we obtain
the moments of the Gaussian ensemble by taking $b=a \to +\infty$ in the general formula
(\ref{momentsJ1})-(\ref{momentsJ2}). This limit is not simple. First, raising
(\ref{zy}) to the power $z^k$ requires adding contributions of various
moments $<y^p>$ of degree $p \leq k$. Second, in each
such moment there is no
obvious term by term simplification in the large $a=b$ limit.
As one sees on (\ref{momentsJ1})-(\ref{momentsJ2}), each term $A_\lambda$ is
superficially of order one, hence multiple cancellations do occur in the sum
so that the end result is of order $1/a^{k/2}$ at large $a$.
The calculation is thus handled using Mathematica, which allows to obtain
moments to high degrees.

Setting $n=0$, we then obtain the first non-trivial cumulants of $\overline{p_\beta(z)}$.
Note that the weight factor is now $e^{- \beta z^2/2}$ hence the disordered model
corresponds to a particle in a LCGP in presence of a quadratic confining background
potential $V_0(z) = z^2/2$ at inverse temperature $\beta$. We obtain
\bea
&& \overline{<z^2>_\beta} = \beta + \beta^{-1}  \\
&& \overline{<z^4>_\beta}^c  = \overline{<z^4>_\beta} - 3 \overline{<z^2>_\beta}^2 = - 1
\eea
and we list here the next three ones i.e. $\overline{<z^k>_\beta}^c$ for
$k=6,8,10$
\bea
\fl ~~~~  \left\{2 \left(\beta +\beta^{-1} \right),-2  \left(3 \beta
   ^2+13 +3 \beta^{-2} \right),12  \left(\beta+\beta^{-1} \right) \left(2 \beta ^2+23 +2 \beta^{-2} \right)\right\}
\eea
These expressions are identical to the ones calculated in
\cite{FLDR10} using there a much more painstaking method.
They are manifestly duality-invariant and freeze at $\beta=1$,
from which one can read the expressions for the corresponding
moments of the position of the minimum $z_m$ (which
we do not write here in detail). Since, as discussed there,
the Burgers velocity $v$ in the inviscid limit is {\it equal} to the
position of the minimum $v \equiv z_m$ in the Gaussian
ensemble problem, this leads to non-trivial predictions for the
moments of the PDF of these two quantities, some of which were
numerically checked there.

\section{Discussion and Conclusions}

\subsection{Numerical verification}

We now compare our predictions for the position $x_m$ of the maximum of the
GUE-CP with the results of direct numerical simulations of
GUE polynomials for matrices of growing size $N$, performed by Nick Simm,
who we also thank for the detailed analysis of the data.

  \begin{figure}[htpb]
  \centering
  \includegraphics[width=0.8\textwidth]{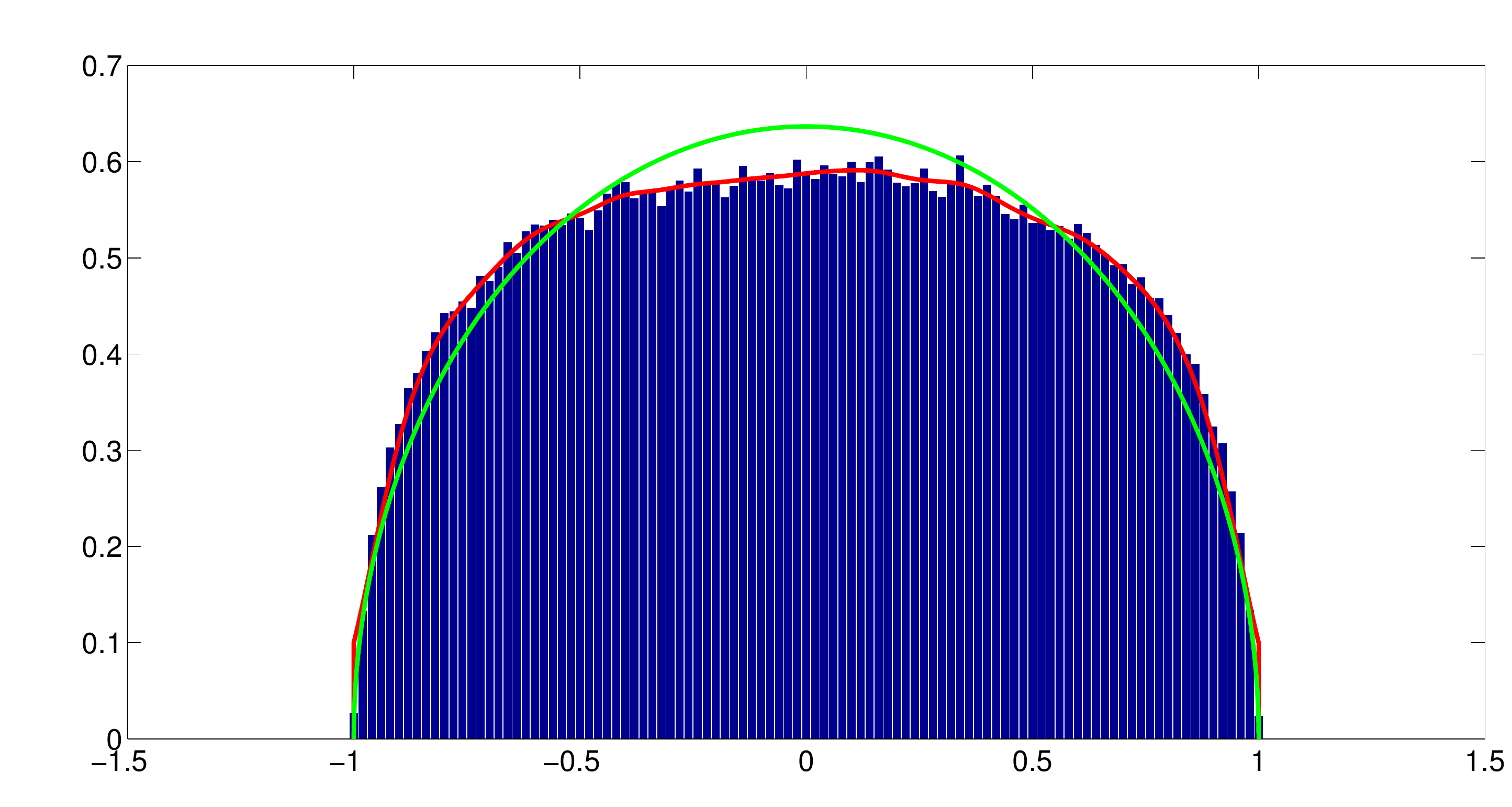}
  \caption{Histogram of values $x_{m}$ for the position of the maximum of the characteristic
  polynomial for size $N=3000$ GUE matrices with $250,000$ realizations. We use the numerical method described in Section $3$ of \cite{FyoSim15}. The curve fitting the histogram (red) differs from the semi-circular density (green) at most by 0.099.}
\label{fig:fig2}
\end{figure}

In Figure \ref{fig:fig2} we show the histogram of the full PDF of the position $x_m$.
For the sake of comparison we also plotted the semi-circle density of eigenvalues, which
as discussed above, has the same first negative moment. The data suggest that, although the distribution of the maximum if clearly not given by the semi-circular law, the edge behaviors of the two distributions are
numerically close.

Next, in Figures \ref{fig:var}, \ref{fig:kurt} we are plotting the values of the second moment and of the kurtosis
as compared to the {\bf Prediction 1}. While the detailed analysis of finite size
effects is left for the future, we plotted the data against the finite size scale
$1/[10 (\ln N)^3]$, which we found appropriate.

We see a rather good agreement for the extrapolated values of the second moment, and still reasonable agreement for the kurtosis.
In Figure \ref{fig:inv}, we also plot the first inverse moment, which shows rather good convergence to the
predicted value $2$. The second inverse moment, predicted to diverge,
is also shown in Figure \ref{fig:inv2}.

In conclusion the agreement with the predictions is reasonable, and for some observables, excellent.
We hope that increasing both the number of realizations and the value of the parameter $N$ should lead to further improvement,
but such a programme is challenging computationally and is left for future research.

\begin{figure}[htpb]
  \centering
  \includegraphics[width=0.8\textwidth]{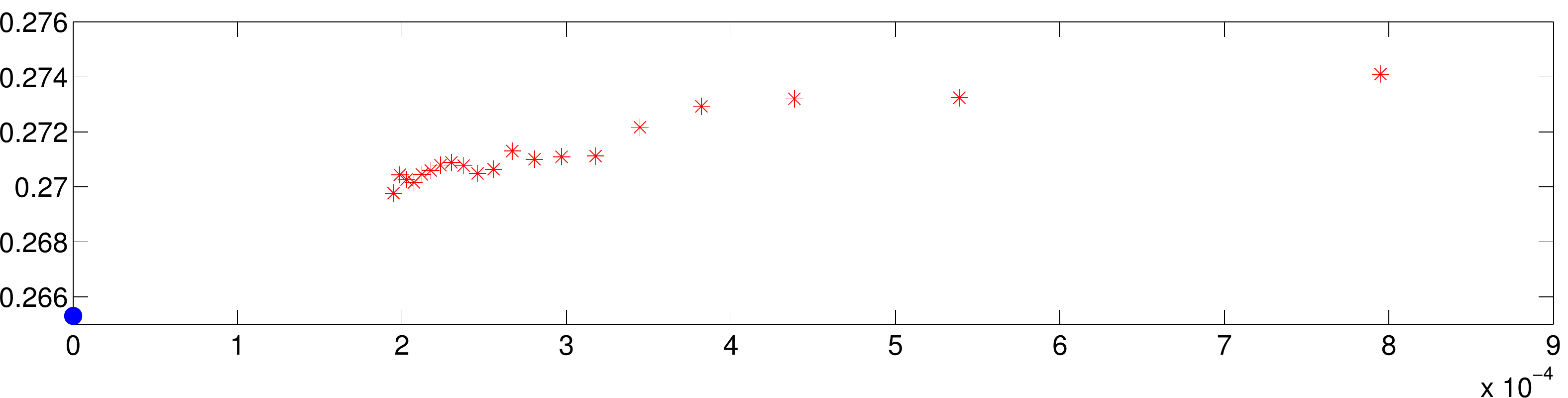}
  \caption{Variance of the position of the maximum of the characteristic
  polynomial for 20 equally spaced data points corresponding to size $N=150$ up to size $N=3000$
  GUE matrices with $250,000$ realizations. The $x$ axis has been chosen as $1/[10 (\ln N)^3]$.
  The blue point is the prediction (\ref{var})}
\label{fig:var}
\end{figure}

\begin{figure}[htpb]
  \centering
  \includegraphics[width=0.8\textwidth]{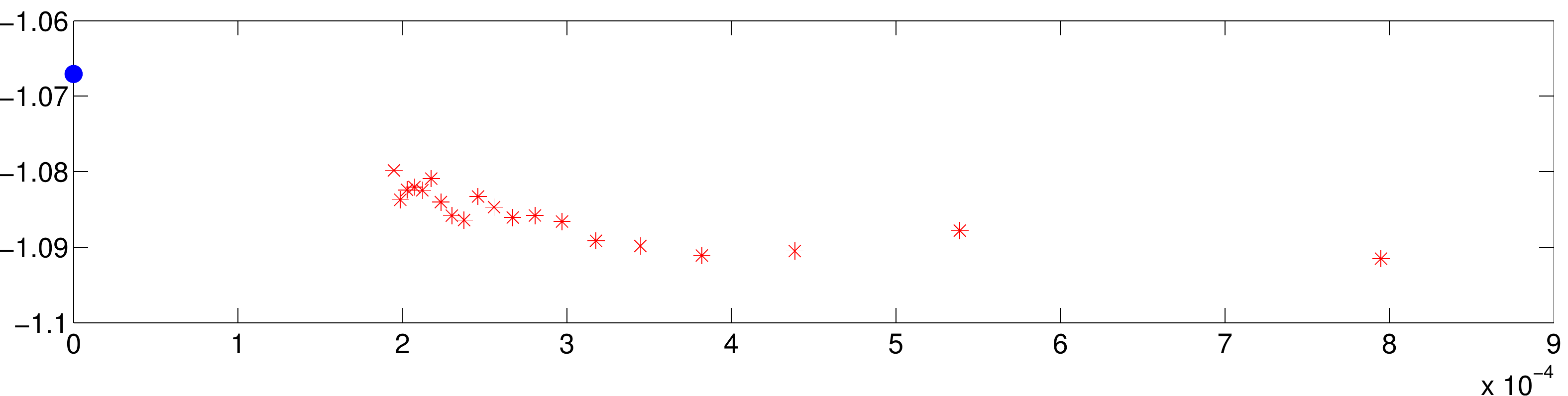}
  \caption{Kurtosis of the position of the maximum of the characteristic
  polynomial, same samples and $x$ axis scale as in Fig.\ref{fig:var}.  The blue point is the
  {\bf Prediction 1}.}
\label{fig:kurt}
\end{figure}

\begin{figure}[htpb]
  \centering
  \includegraphics[width=0.8\textwidth]{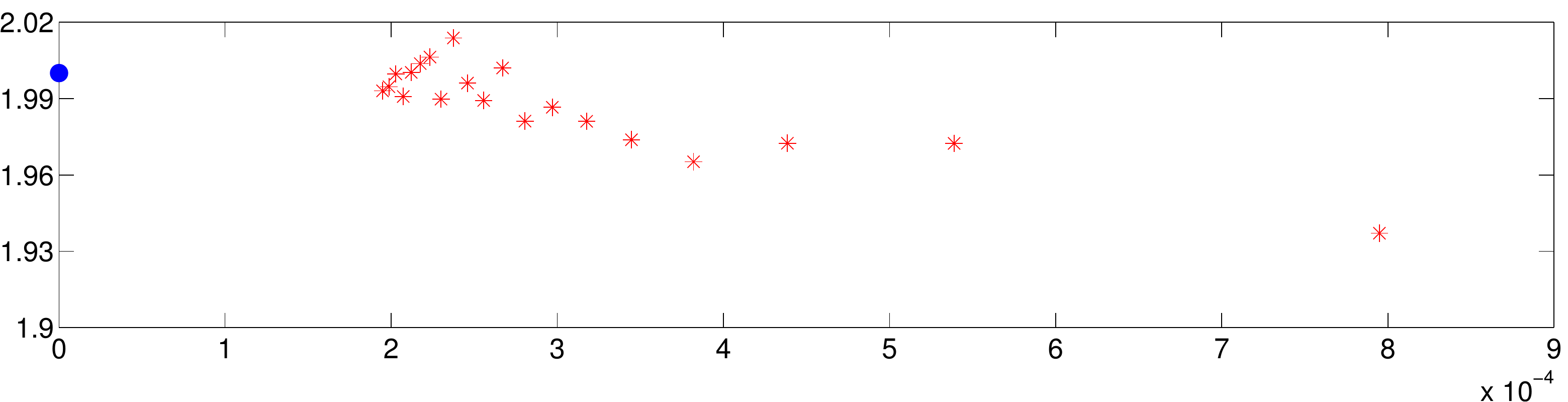}
  \caption{Inverse moment of the position of the maximum of the characteristic
  polynomial, same samples and $x$ axis scale as  in Fig.\ref{fig:var}. The blue point is the prediction
 (\ref{predinv}).}
\label{fig:inv}
\end{figure}

 \begin{figure}[htpb]
  \centering
  \includegraphics[width=0.8\textwidth]{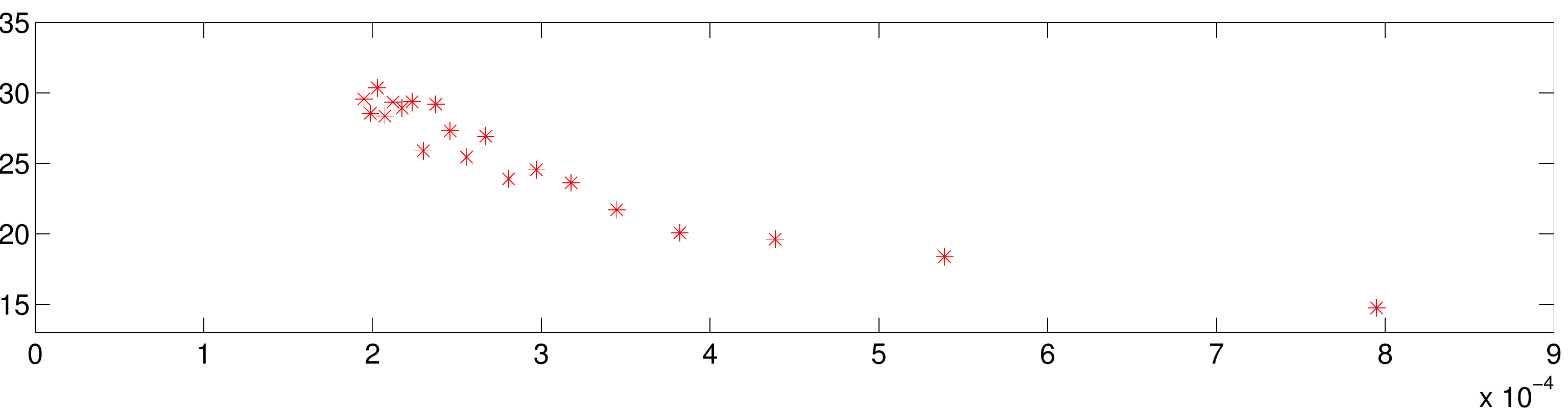}
  \caption{Inverse moment of the position of the maximum of the characteristic
  polynomial, same samples and $x$ axis scale as in Fig.\ref{fig:var}.
  The prediction is a divergence of the moment as $N \to +\infty$.}
\label{fig:inv2}
\end{figure}

\subsection{Conclusions}

In this paper we developed a systematic approach to investigating statistical properties
of the position $x_m$ of the global extremum (maximum or minimum) for appropriately
regularized logarithmically-correlated gaussian (LCG) processes in an interval.
We explicitly treat three processes of that kind: the logarithm of the Gaussian unitary ensemble (GUE) characteristic polynomial, the log-correlated potential in presence
of edge charges, and the Fractional Brownian motion with Hurst index $H \to 0$,
The distribution of $x_m$ is characterized through its positive integer moments $\mathbb{E}\left\{\left[x_m\right]^k\right\}$  for which we provided an explicit, closed form expression in terms of a sum over partitions of the integer $k$. Our approach is based on the idea of interpreting LCGs as a random potential so that the associated Boltzmann-Gibbs measure in the limit of zero temperature $T \to 0$ concentrates around the coordinates of the global minimum for the potential.
 Our main technical instrument of analysis is then combining the replica trick representation for the Boltzmann-Gibbs average of $\left[x_m\right]^k$ with the possibility of exact evaluation of that average by mapping it to the moment problem for $\beta-$Jacobi ensemble of random matrices. To perform the latter we used a method based on Jack polynomials expansion and Macdonald-Kadell integral, with alternative routes possible via Borodin-Gorin moment formula. The latter appproach provides also some expressions for negative integer moments.
  Our calculations yield explicit formulae for the moments in the high-temperature phase $T>T_c$, which can be further continued to $T<T_c$ by exploiting the Freezing-Duality Conjecture, and in this way provide for $T\to 0$ the sought for expressions for $\mathbb{E}\left\{\left[x_m\right]^k\right\}$. Although any integer moment of $x_m$ can be with due effort calculated in that way, it remains a challenge to convert such information into the appropriate generating function for the distribution of $x_m$. Doing this would allows to understand, e.g., far tails of the latter distribution, and we leave that and other interesting question for the future work.
 Note that we have also provided results about correlations between position
 and value of the maximum, and a general method to calculate conditional moments. The determination of the full the joint probability is also left for the future. We also used the numerical data provided by Nick Simm to test our predictions for the moments of the position for the maximum in GUE case.
 The agreement with our theory is reasonable, and for some observables, excellent.

 The application of our method to a non-stationary process, the fBm0, poses several
questions. While the analytical continuation to $n=0$ seems to be benign for the
moments of $x_m$, and leads to the interesting {\bf Prediction 3}, extension to
determine the PDF of the global minimum seems to fail, as discussed in the
Appendix G. It would thus be desirable in near future to find a way around this problem, as
well as to check the predictions for the moments against numerical simulations of the process $B^{(\eta )}_0(x)$. Unfortunately generating many instances with reliable precision seems to be extremely time-demanding. To this end recall that there exists an intimate relation between $B^{(\eta )}_0(x)$ and the behaviour of the ({\it increments} of)  GUE characteristic polynomials, though at a different, so-called "mesoscopic" spectral scales \cite{FKS13}, negligible in comparision with the interval $[-1,1]$.
Clearly, such restriction makes the problem challenging and we leave numerical verification/falsification of the {\bf Prediction 3} for the future.

 We believe that our analysis of $\beta-$Jacobi ensemble is interesting in its own right and complements one that has appeared in the very recent work by Mezzadri and Reynolds \cite{MR15}.
 We have also presented conjectures for moments in related ensemble.


\section{Appendix A: contour integral formulas for Jacobi ensemble}

{\it Appendix written by Alexei Borodin and Vadim Gorin}

\medskip

Consider the $N$--particle Jacobi ensemble, which is a probability distribution on
$N$--tuples of reals $0\le \r_1\le \r_2\le\dots\le \r_N\le 1$ with density
\begin{multline}
\label{eq_general_beta_Jacobi}
 \Pr^{\alpha,M,\theta}(\r\in [h,h+dh])= {\rm const}\cdot
  \prod_{1\le i<j \le N} (h_j-h_i)^{2\theta}
 \prod_{i=1}^{N} h_i^{\theta\alpha-1} (1-h_i)^{\theta(M-N+1)-1} d h_i,
\end{multline}
where $M\ge N$ is an integer, and $\alpha>0$, $\theta>0$ are two real parameters.
$2\theta$ is customary called $\beta$ in the random matrix theory.

\begin{theorem} \label{Theorem_positive_moment}
For $k=1,2,\dots$ the expectation of $\sum_{i=1}^N (\r_i)^{k}$ with respect to the
measure \eqref{eq_general_beta_Jacobi} is given by
\begin{multline}
\label{eq_expectation_of_moment_positive} \frac{(-\theta)^{-1}}{(2\pi\ii)^k }
\oint\dots\oint \dfrac{1}{\left(u_2-u_1+1-\theta\right)\cdots \left(u_k-u_{k-1}+1-\theta\right)}\\
\times \prod_{i<j}
 \frac{\left(u_j-u_i\right)\left(u_j-u_i+1-\theta\right)}{\left(u_j-u_i-\theta\right)\left(u_j-u_i+1\right)}
 \left(\prod_{i=1}^k \frac{u_i-\theta}{u_i+(N-1)\theta}  \cdot
\frac{u_i-\theta\alpha}{u_i-\theta\alpha-\theta M} du_i \right),
\end{multline}
where all the contours enclose singularities at $(1-N)\theta$ and not at
$\theta\alpha+\theta M$, $|u_1|\ll|u_2|\ll\dots\ll|u_k|$.
\end{theorem}

\begin{theorem} \label{Theorem_negative_moment}
For any positive integer $k<\theta\alpha$ the expectation of $\sum_{i=1}^N
(\r_i)^{-k}$ with respect to the measure \eqref{eq_general_beta_Jacobi} is given by
\begin{multline}
\label{eq_expectation_of_moment_negative} \frac{\theta^{-1}}{(2\pi\ii)^k }
\oint\dots\oint \dfrac{1}{\left(u_2-u_1-1+\theta\right)\cdots \left(u_k-u_{k-1}-1+\theta\right)}\\
\times \prod_{i<j}
 \frac{\left(u_j-u_i\right)\left(u_j-u_i-1+\theta\right)}{\left(u_j-u_i+\theta\right)\left(u_j-u_i-1\right)}
 \left(\prod_{i=1}^k \frac{u_i+N\theta}{u_i}  \cdot
\frac{u_i+1-\theta\alpha-M\theta}{u_i+1-\theta\alpha} du_i \right),
\end{multline}
where all the contours enclose singularities at $0$ and not at $\theta\alpha-1$,
$|u_1|\ll|u_2|\ll\dots\ll|u_k|$.
\end{theorem}
\begin{remark}
 There exist similar contour integral formulas for the expectations of the powers $(\sum_{i=1}^N
 (\r_i)^{k})^m$ with $m=1,2,\dots$ and $k$ being both positive and negative
 integers. They are obtained by \emph{iterating} the results of Propositions
 \ref{Prop_op_negative}, \ref{Prop_op_positive} below.
\end{remark}

In the rest of this section we prove Theorems \ref{Theorem_positive_moment} and
\ref{Theorem_negative_moment}.

The starting point of our proof is a formula from \cite{Negut} which we now present.
Let $\Lambda$ be the algebra of symmetric polynomials in infinitely many variables
$x_1,x_2,\dots$. This algebra is naturally identified with polynomial algebra
$\mathbb C[p_1,p_2,\dots]$, where $p_k$ are Newton power sums:
$$
 p_k=x_1^k+x_2^k+x_3^k+\dots,\quad k=1,2,\dots.
$$
We also need a distinguished linear basis in $\Lambda$ consisting of Macdonald
polynomials $P_\lambda(\cdot;q,t)$, which depend on two parameters $q$ and $t$; here
and below we use the notations of \cite{Macdonald} and $\lambda=\lambda_1\ge
\lambda_2\ge \dots\ge 0$ is a Young diagram.

\begin{proposition}
\label{Proposition_correct_definition_2} Define a differential operator ${\mathfrak
D}_n$ acting in $\Lambda$ via
\begin{multline}
\label{eq_cor_def_d}
   {\mathfrak D}_n = \frac{(-1)^{n-1}}{(2\pi\ii)^n } \oint\dots \oint \dfrac{\sum_{i=1}^n
 \frac{z_n}{z_i (t/q)^{n-i}}}{\left(1-\frac{qz_2}{t z_1}\right)\cdots \left(1-\frac{qz_n}{t
 z_{n-1}}\right)} \prod_{i<j}
 \frac{\left(1-\frac{z_i}{z_j}\right)\left(1-\frac{tz_i}{qz_j}\right)}{\left(1-t\frac{z_i}{z_j}\right)\left(1-q^{-1}\frac{z_i}{z_j}\right)}
\\ \times \exp\left(\sum_{k=1}^\infty q^{-k}(1-t^{k})
\frac{z_1^{-k}+\dots+z_n^{-k}}{k}p_k\right)\exp\left(\sum_{k=1}^\infty
\frac{z_1^{k}+\dots+z_n^{k}}{k}
 (1-q^k) \frac{\partial}{\partial p_k}\right)
\\ \times \frac{dz_1}{z_1}\cdots \frac{dz_n}{z_n},
\end{multline}
where the contours are circles around $0$ satisfying
$|z_1|\ll|z_2|\ll\dots\ll|z_n|$,  and $p_k$ means the operator of multiplication by
$p_k$. Then Macdonald polynomials are eigenfunctions of ${\mathfrak D}_n$, i.e.\
$${\mathfrak D}_n P_\lambda(\cdot;q,t)=(1-t^n) \sum_{i=1}^\infty (q^{-\lambda_i}
t^{i-1})^n P_\lambda(\cdot;q,t).
$$
\end{proposition}
\begin{proof} In slightly different notations this is \cite[Theorem
1.2]{Negut}.\end{proof}

Next, we set $x_{N+1},x_{N+2},\dots$ equal to $0$ and apply the above operators to
an $N$--variable product function $f(x_1)\cdots f(x_N)$ (which belongs to the space
of symmetric power series in $x_1,\dots,x_N$). Then we get

\begin{multline}
\label{eq_op_finite}
   {\mathfrak D}_n \prod_{i=1}^N f(x_i) = \left(\prod_{i=1}^N f(x_i)\right) \frac{(-1)^{n-1}}{(2\pi\ii)^n } \oint\dots \oint \dfrac{\sum_{i=1}^n
 \frac{z_n}{z_i (t/q)^{n-i}}}{\left(1-\frac{qz_2}{t z_1}\right)\cdots \left(1-\frac{qz_n}{t
 z_{n-1}}\right)}\\ \times \prod_{i<j}
 \frac{\left(1-\frac{z_i}{z_j}\right)\left(1-\frac{tz_i}{qz_j}\right)}{\left(1-t\frac{z_i}{z_j}\right)\left(1-q^{-1}\frac{z_i}{z_j}\right)}
 \left(\prod_{i=1}^n \prod_{a=1}^N \frac{z_i-tq^{-1} x_a}{z_i-q^{-1}x_a}\right)  \prod_{i=1}^n
\frac{f(z_i)}{f(qz_i)} \frac{dz_i}{z_i}.
\end{multline}
In the last formula the contours are \emph{large circles} (this is because the
integrand needs to be decomposable into a symmetric power series in the variables
$x_i$ to justify the computation).

At this point we can pass from formal point of view based on the algebra $\Lambda$
to the analytic one and view $q$, $t$ and $x_i$ as real (or complex) numbers. Our
next step is the following limit transition:
\begin{equation}
\label{eq_limit_transition}
 \eps\to 0,\quad q=\exp(-\eps),\quad t=q^{\theta},\quad z_i=\exp(\eps u_i),\quad x_i=\exp(\eps  y_i),\quad \lambda_i=\eps^{-1} r_i.
\end{equation}
The Macdonald polynomials $P_\lambda$ are shown in \cite{BG_HO} to converge to the
Heckman--Opdam hypergeometric functions $HO_r(y_1,\dots,y_N;\theta)$ in this limit
regime:
\begin{equation}\label{eq_P_limit}
 \lim_{\eps\to 0} P_\lambda(x_1,\dots,x_N;q,t)= HO_r(y_1,\dots,y_N;\theta).
\end{equation}

In order to pass to the limit $\eps\to 0$ in the operators $ {\mathfrak D}_n$ we
need to deform the integration contours so that they enclose the singularities in
$q^{-1}x_i$ but not $0$. Due to the ordering of the contours we should do this one
after another: first deform the $z_1$ contour, then deform the $z_2$ contour, etc.
In principle, in this process we get $2^n$ terms obtained by taking the residues at
$0$ in a subset of variables $z_1,\dots,z_n$.

\begin{lemma}
\label{lemma_residues}
 Only subsets of the form $z_k,z_{k+1},\dots,z_\ell$, $0\le k<\le l \le n$ give non-zero residues in
 \eqref{eq_op_finite}.
\end{lemma}
\begin{proof}
 Recall that we should compute the residues sequentially, from $z_1$ to $z_n$. Suppose that we
 started from $z_k$. Then the residue is $n-1$ dimensional integral, in which the factors in the
 second line of \eqref{eq_op_finite} are still the same (with additional prefactor $t^N$), while the factor in the first line transforms into
\begin{equation}
\label{eq_x4} (-1)\frac{\frac{z_n}{z_{k+1} (t/q)^{n-k-1}} } {\left(1-\frac{qz_2}{t
z_1}\right)\cdots \left(1-\frac{qz_{k-1}}{t
 z_{k-2}}\right)     \left(1-\frac{qz_{k+2}}{t z_{k+1}}\right)\cdots \left(1-\frac{qz_n}{t
 z_{n-1}}\right)}
\end{equation}
Note that if at the next step we do not take the residue in $z_{k+1}$, then all
further residues in $z_i$, $i>k+1$ will be zero --- the integrand will simply have
no poles at $0$ in these variables. If we take the residue in $z_{k+1}$, then we get
\begin{equation}
\label{eq_x5}(-1)^{2}\frac{\frac{z_n}{z_{k+2} (t/q)^{n-k-2}} }
{\left(1-\frac{qz_2}{t z_1}\right)\cdots \left(1-\frac{qz_{k-1}}{t
 z_{k-2}}\right)     \left(1-\frac{qz_{k+3}}{t z_{k+2}}\right)\cdots \left(1-\frac{qz_n}{t
 z_{n-1}}\right)}
\end{equation} which still has the same form and, thus, we can continue in the same way.
\end{proof}

In particular, if we take the residue at $0$ with respect to all variables
$z_1,\dots,z_n$, then we get
$$
 (-1)^{n-1} t^{Nn}.
$$
Note that $(-1)^{n-1}$ cancels out with the integral prefactor. Let us pass to the
operator
$$
  {\mathfrak D}_n-t^{Nn}
$$
On one hand, this operator is given by the expansion into $n(n+1)/2$ integrals of
various dimensions integrated around $q^{-1}x_i$ but not $0$. On the other hand, its
eigenvalues on Macdonald polynomials in $N$ variables are
$$
 (1-t^n) \left(\sum_{i=1}^N (q^{-\lambda_i} t^{i-1})^n + \sum_{i=N+1}^N (t^{i-1})^n\right) -t^{Nn}=  (1-t^n) \sum_{i=1}^N (q^{-\lambda_i} t^{i-1})^n
$$
The eigenvalues of $\eps^{-1} ({\mathfrak D}_n-t^{Nn})$ converge in the limit regime
\eqref{eq_limit_transition} to
$$
 \theta n \sum_{i=1}^N \exp(n r_i).
$$
Thus, taking into account \eqref{eq_P_limit}, we conclude that integral
representation for $\eps^{-1} ({\mathfrak D}_n-t^{Nn})$ should also converge. The
$n$--dimensional integral here converges to
\begin{multline}
\label{eq_ci_finite_limit}
   \left(\prod_{i=1}^N f(y_i)\right) \frac{(-1)^{n-1}}{(2\pi\ii)^n } \oint\oint \dfrac{n}{\left(u_1 -u_2+1-\theta\right)\cdots \left(u_{n-1}-u_n+1-\theta\right)}\\ \times \prod_{i<j}
 \frac{\left(u_j-u_i\right)\left(u_j-u_i-1+\theta\right)}{\left(u_j-u_i+\theta\right)\left(u_j-u_i-1\right)}
 \left(\prod_{i=1}^n \prod_{a=1}^N \frac{u_i-y_a+\theta-1}{u_i-y_a-1}\right)  \prod_{i=1}^n
\frac{f(u_i)}{f(u_i-1)}du_i,
\end{multline}
where all the contours enclose singularities at $y_a+1$ and
$|u_1|\ll|u_2|\ll\dots\ll|u_n|$.

Note that is was important to have $n-1$ factors in the first line. Indeed, each
such factor produced $\eps^{-1}$. On the other hand, $\eps^n$ was produced by the
change of variables. Together with additional $\eps^{-1}$ in the definition of our
limit transition this gave precisely the constant order as $\eps\to 0$. Now note
that when we computed the residues as in Lemma \ref{lemma_residues}, then when
$1<k<n$, the $m$--dimensional integral would come with less than $m-1$ factors in
the first line. Indeed, this is clearly visible in \eqref{eq_x4}, \eqref{eq_x5}. If
follows that such terms \emph{vanish} in our limit transition. Therefore, only terms
with $k=1$ or $k=n$ survive. A general term is obtained either by taking the residue
in variables
$$
 z_1,\dots,z_{\ell},\quad \ell=1,\dots,n-1
$$
as in Lemma \ref{lemma_residues} and then sending $\eps\to 0$.
 As a result, we get $n-\ell$
dimensional integral
\begin{multline}
\label{eq_ci_finite_limit_2}
   \left(\prod_{i=1}^N f(y_i)\right) \frac{(-1)^{n-1-\ell}}{(2\pi\ii)^{n-\ell} }
   \oint\oint \dfrac{1}{\left(u_{\ell+1} -u_{\ell+2}+1-\theta\right)\cdots \left(u_{n-1}-u_n+1-\theta\right)}\\ \times \prod_{i<j}
 \frac{\left(u_j-u_i\right)\left(u_j-u_i-1+\theta\right)}{\left(u_j-u_i+\theta\right)\left(u_j-u_i-1\right)}
 \left(\prod_{i=\ell+1}^n \prod_{a=1}^N \frac{u_i-y_a+\theta-1}{u_i-y_a-1}\right)  \prod_{i=\ell+1}^n
\frac{f(u_i)}{f(u_i-1)}du_i,
\end{multline}
Or we can take the residue in variables
$$
 z_k,\dots,z_{n},\quad,  n=1,\dots,n-1
$$
as in Lemma \ref{lemma_residues} and then send $\eps\to 0$. Note that under the
identification $k=n+1-\ell$, the integrals have \emph{the same} integrand. However,
the signs appearing when we take residues are different. Namely, when we take
residues in $z_1,\dots,z_\ell$ we get the sign $(-1)^{\ell}$. On the other hand,
when we take residues in $z_{n+1-\ell},\dots,z_n$ the sign is $(-1)^{\ell-1}$, since
we do not get $(-1)$ factor at the very last step. As a conclusion, two such terms
precisely cancel out.

Shifting the variables $u\mapsto u+1$ and dividing by $\theta n$ we write the final
formula.

\begin{proposition} \label{Prop_op_negative} The action of the operator $\mathcal P_n:=\frac{1}{\theta n}\lim_{\eps\to 0} \eps^{-1} ({\mathfrak D}_n-t^{Nn})$
with eigenvalues
$$
 \mathcal P_n HO_r(y_1,\dots,y_N;\theta)=   \sum_{i=1}^N \exp(n r_i) HO_r(y_1,\dots,y_N;\theta),
$$
on a function $f(y_1)\cdots f(y_N)$ can be computed via
\begin{multline} \label{eq_Operator_limit_integral}
\frac{\mathcal P_n \prod_{i=1}^N f(y_i)}{\prod_{i=1}^N f(y_i)} =
\frac{\theta^{-1}}{(2\pi\ii)^n }
\oint\dots\oint \dfrac{1}{\left(u_2 -u_1-1+\theta\right)\cdots \left(u_n-u_{n-1}-1+\theta\right)}\\
\times \prod_{i<j}
 \frac{\left(u_j-u_i\right)\left(u_j-u_i-1+\theta\right)}{\left(u_j-u_i+\theta\right)\left(u_j-u_i-1\right)}
 \left(\prod_{i=1}^n \prod_{a=1}^N \frac{u_i-y_a+\theta}{u_i-y_a}\right)  \prod_{i=1}^n
\frac{f(u_i+1)}{f(u_i)}du_i
\end{multline}
where all the contours enclose singularities at $y_a$, $f(u+1)/f(u)$ is analytic
inside the contours and $|u_1|\ll|u_2|\ll\dots\ll|u_n|$.
\end{proposition}

Further, note that in the operator ${\mathfrak D}_n-t^{Nn}$ we can freely invert the
variables $(q,t)\mapsto (q^{-1},t^{-1})$. Since, the Macdonald polynomials are
invariant under this change (cf.\ \cite{Macdonald}), they are still eigenfunctions.
Moreover, since we never used the fact $0<t,q<1$ in the proofs, all the integral
representations are still valid. Thus, we arrive at:
\begin{proposition} \label{Prop_op_positive} The action of the operator $\widehat {\mathcal P}_n$ with eigenvalues
$$
\widehat {\mathcal P}_n HO_r(y_1,\dots,y_N;\theta)=   \sum_{i=1}^N \exp(-n r_i)
HO_r(y_1,\dots,y_N;\theta),
$$
on a function $f(y_1)\cdots f(y_N)$ can be computed via
\begin{multline*}
\frac{\mathcal P_n \prod_{i=1}^N f(y_i)}{\prod_{i=1}^N f(y_i)} =
\frac{(-\theta)^{-1}}{(2\pi\ii)^n }
\oint\dots\oint \dfrac{1}{\left(u_2-u_1+1-\theta\right)\cdots \left(u_n-u_{n-1}+1-\theta\right)}\\
\times \prod_{i<j}
 \frac{\left(u_j-u_i\right)\left(u_j-u_i+1-\theta\right)}{\left(u_j-u_i-\theta\right)\left(u_j-u_i+1\right)}
 \left(\prod_{i=1}^n \prod_{a=1}^N \frac{u_i-y_a-\theta}{u_i-y_a}\right)  \prod_{i=1}^n
\frac{f(u_i-1)}{f(u_i)}du_i
\end{multline*}
where all the contours enclose singularities at $y_a$, $f(u-1)/f(u)$ is analytic
inside the contours and $|u_1|\ll|u_2|\ll\dots\ll|u_n|$.
\end{proposition}

Now we are ready to prove the contour integral formulas for the Jacobi ensemble.

\begin{proof}[Proof of Theorem \ref{Theorem_negative_moment}] This is essentially a
corollary of Proposition \ref{Prop_op_negative} and the results of \cite{BG_HO} and
below we sketch the proof omitting some technical details, cf.\ \cite[Section
2.3]{BG_HO} for similar arguments.

The Cauchy identity for the Macdonald polynomials yields for $M\ge N$
\begin{equation} \label{eq_Cauchy}
 \sum_{\lambda=(\lambda_1\ge\lambda_2\ge\dots\ge\lambda_N\ge 0)}
 \frac{P_\lambda(x_1,\dots,x_N;q,t)
 P_\lambda(\hat x_1,\dots,\hat x_M;q,t)}{\langle P_\lambda,P_\lambda\rangle}=\prod_{i=1}^N
 \prod_{j=1}^{M} \frac{ (tx_i\hat x_j;q)_\infty}{(x_i \hat x_j;q)_\infty},
\end{equation}
where $(a;q)_\infty$ is the $q$--Pochhammer symbol,
$(a;q)_\infty=\prod_{\ell=0}^{\infty} (1-a q^{\ell})$, and $\langle
P_\lambda,P_\lambda\rangle$ are certain explicit constants, which can be found e.g.\
in \cite[Chapter VI]{Macdonald}. We further do the following three steps:
\begin{enumerate}
 \item Apply $\eps^{-1} ({\mathfrak D}_k-t^{Nk})$ to both sides of
 \eqref{eq_Cauchy} and then divide them by $\prod_{i=1}^N
 \prod_{j=1}^{M} \frac{ (tx_i\hat x_j;q)_\infty}{(x_i \hat x_j;q)_\infty}$.
 \item Set $x_i=t^{i-1}$ and $\hat x_j = t^{\alpha} t^{j-1}$ and send $\eps\to 0$ in the limit regime \eqref{eq_limit_transition}.
 \item Evaluate the $\eps\to 0$ limit of both sides using Proposition
 \ref{Prop_op_negative}.
\end{enumerate}
We claim that the resulting identity is precisely the statement of Theorem
\ref{Theorem_negative_moment}. Indeed, it is shown in \cite[Theorem 2.8]{BG_HO} that
as $\eps\to 0$,
$$
\left( \prod_{i=1}^N
 \prod_{j=1}^{M} \frac{ (tx_i\hat x_j;q)_\infty}{(x_i \hat
 x_j;q)_\infty}\right)^{-1}
\frac{P_\lambda(x_1,\dots,x_N;q,t)
 P_\lambda(\hat x_1,\dots,\hat x_M;q,t)}{\langle P_\lambda,P_\lambda\rangle}
 $$
 converges to  the density of the Jacobi ensemble
 \eqref{eq_general_beta_Jacobi} in variables $\r_i=\exp(-r_i)$. Together with the eigenrelation for $\mathcal P_n$
 of Proposition \ref{Prop_op_negative} this implies that the left--hand side
 of the identity is the expectation of $\sum_{i=1}^N (\r_i)^{-k}$. For the right--hand side
 we use the convergence of the $q$--Pochhammer symbols to the Gamma function as
 $q\to 1$ to get the expression of the form \eqref{eq_Operator_limit_integral} with
 $$
  f(y)=\frac{\Gamma(-y+\theta \alpha)}{\Gamma(-y+\theta \alpha +\theta M)},
 $$
 which is precisely \eqref{eq_expectation_of_moment_negative}.
\end{proof}

Theorem \ref{Theorem_positive_moment} is proven in the same way as Theorem
\ref{Theorem_negative_moment}, but using Proposition \ref{Prop_op_positive} instead
of Proposition \ref{Prop_op_negative}.

\section{ Appendix B: calculation of contour integrals and more results for moments}

In this Appendix we give some details of the contour integral calculations,
as well as some additional explicit results. All formula presented here
have also been obtained our expressions of moments in terms of partitions.

\subsection{Second moment} \label{app:second}

\noindent

Consider
(\ref{contourbeta}) for $k=2$. One first perform the contour integral over $u_1$
on a small circle around $0$.
The poles in $u_1$ are at $u_1=0,2+a+b+2 \beta^2,1+u_2,\beta^2+u_2$.
Because of condition $C_1$ the last two poles are not encountered.
Because of $C_2$ neither is the second pole. Hence only the pole
in $u_1=0$ is picked up and from its residue one gets the remaining integral:
\bea
\fl <y^2>_{\beta,a,b,n=0} = \frac{\left(a+\beta ^2+1\right)}{ \left(a+2 \beta ^2+b+2\right)}
\int \frac{du_2}{2 i \pi}  \frac{u_2 \left(u_2- (a+\beta^2+1) \right)}
{\left(u_2+1\right)  \left(u_2+ \beta^2\right) \left(u_2 - (a+2 \beta ^2+2) \right)}
\eea
From the prescription $C_2$ only the poles at $u_2=-1$ and $u_2=-\beta^2$
contribute and, summing their residues one obtains (\ref{y22}) in the text.

For completeness we now give the complete $n$ dependence of several
results in the text.

For the LCP with edge charges in the case $b=a$ one finds:
\bea
\fl && < y^2 >_{\beta,a,a,n} =
\frac{2 a^2+a \left(\beta ^2 (6-5 n)+6\right)+\beta ^4 (n-1) (3 n-4)+\beta ^2 (9-7
   n)+4}{2 \left(2 a+\beta ^2 (3-2 n)+2\right) \left(2 a-2 \beta ^2 (n-1)+3\right)}
\eea

For the fBm, with $a= 2 n \beta^2$, $b=0$ one has
\bea \label{fBmn}
< y^2 >_{\beta,n} =
\frac{\left(\beta ^2 (n+1)+1\right) \left(\beta ^4 \left(n^2+n+4\right)+\beta ^2
   (n+9)+4\right)}{2 \left(6 \beta ^6+19 \beta ^4+19 \beta ^2+6\right)}
\eea

For the GUE-CP, setting $a=b=\frac{1+\beta^2}{2}$, one finds:
\bea \label{GUEn}
<y^2>_{\beta,n}  = \frac{\beta ^4 (2 n-3) (3 n-5)+\beta ^2 (32-19 n)+15}{4 \left(2 \beta ^2 (n-2)-3\right)
   \left(\beta ^2 (2 n-3)-4\right)}
\eea

\subsection{Third moment}

Consider
(\ref{contourbeta}) for $k=3$.
We take successively the residue at $u_1=0$, then
at $u_2=-1,-\beta^2$, which produces two terms,
then at $u_3=-2,-\beta^2$ for the first term, and
at $u_3=-1,- 2 \beta^2$ for the second term.
The final result for arbitrary $a,b,\beta$ is too heavy to
reproduce here. We give here:

\begin{itemize}

\item the third cumulant for arbitrary $a,b$, which is slightly simpler:
\bea \label{3cum}
\fl && \overline{ \bigg( y_m - \langle y_m \rangle \bigg)^3 }
=  -\frac{(a+2) (b+2) (a-b) \left(7 a^2+2 a (7 b+34)+b (7 b+68)+164\right)}{(a+b+4)^3 (a+b+5)^2 (a+b+6)^2} \nonumber  \\
\fl &&
\eea
and of course vanishes in the case $a=b$. The associated skewness is
\bea \label{Skew1}
Sk= \frac{(b-a) (a+b+5) \left(7 a^2+2 a (7 b+34)+b (7 b+68)+164\right)}{\sqrt{a+2} \sqrt{b+2}
   (a+b+6)^2 (2 a+2 b+9)^{3/2}}
\eea
to be compared with the skewness for the measure $y^a (1-y)^b$, which is
\bea \label{Skew2}
Sk_0 =-\frac{2 (a-b) \sqrt{a+b+3}}{\sqrt{a+1} \sqrt{b+1} (a+b+4)}
\eea

\item we have checked that the result for the third moment for $a=b$,
and arbitrary $\beta,n$
\bea
< \bigg( y - \frac{1}{2} \bigg)^3 >_{\beta,a,a,n} = 0
\eea
is consistent with the symmetry $y \to 1-y$.

\end{itemize}

\subsection{Fourth moment}
\label{app:4}

Consider now
(\ref{contourbeta}) for $k=4$.
The sequence of poles is the same as for $k=3$
for $u_1,u_2,u_3$,
except for the $u_4$ integration which contains now
four terms and picks poles at $u_4=-3,-\beta^2$;
$u_4=-2,-2 \beta^2,-1-\beta^2$ (twice) and $u_4=-1,-3 \beta^2$.
Clearly this is a simple regular structure which carries on to
higher moments. While the pole structure is simple, we were not
able to find a systematics for the residue valid to any order.

Again, the final result for arbitrary $a,b,\beta$ is too heavy to
reproduce here. We give:

\begin{itemize}

\item the fourth moment, cumulant and kurtosis for the GUE-CP associated
statistical model with $b=a=\frac{1+\beta^2}{2}$,
at arbitrary $\beta \leq 1$:
\bea
\fl && \overline{ <y^4>_\beta } = \frac{252 \beta ^8+1195 \beta ^6+1918 \beta ^4+1195
   \beta ^2+252}{64 \left(\beta ^2+2\right) \left(2
   \beta ^2+1\right) \left(3 \beta ^2+4\right)
   \left(4 \beta ^2+3\right)}  \\
\fl   && \overline{ <x^4>_\beta}  =
  \frac{12 \beta ^8+59 \beta ^6+98 \beta ^4+59 \beta
   ^2+12}{4 \left(\beta ^2+2\right) \left(2 \beta
   ^2+1\right) \left(3 \beta ^2+4\right) \left(4
   \beta ^2+3\right)}   \\
\fl    && \overline{<x^4>_\beta} - 3 \overline{<x^2>_\beta }^2 = -\frac{72 \beta ^{12}+540 \beta ^{10}+1549 \beta
   ^8+2170 \beta ^6+1549 \beta ^4+540 \beta ^2+72}{4
   \left(\beta ^2+2\right) \left(2 \beta ^2+1\right)
   \left(3 \beta ^2+4\right)^2 \left(4 \beta
   ^2+3\right)^2} \nonumber  \\
\fl    && {\rm Ku}  = -\frac{72 \beta ^{12}+540 \beta ^{10}+1549 \beta
   ^8+2170 \beta ^6+1549 \beta ^4+540 \beta ^2+72}{4
   \left(\beta ^2+2\right) \left(2 \beta ^2+1\right)
   \left(3 \beta ^4+7 \beta ^2+3\right)^2}
   \eea
Setting $\beta=1$, this leads to the predictions given in the text for the maximum of the
GUE-CP.\\

\item the fourth moment and cumulants for the statistical model associated to the
fBm $b=a=0$
at arbitrary $\beta \leq 1$:
\bea
\fl && \overline{<y^4>_\beta} = \frac{72 \beta ^8+382 \beta ^6+647 \beta ^4+382 \beta ^2+72}{6
   \left(2 \beta ^2+3\right) \left(2 \beta ^2+5\right) \left(3
   \beta ^2+2\right) \left(5 \beta ^2+2\right)}
\eea
and the fourth cumulant:
\bea
\fl && \overline{\langle \bigg( y - \frac{1}{2} \bigg)^4 \rangle_\beta} - 3 \overline{\langle \bigg( y - \frac{1}{2} \bigg)^2 \rangle_\beta}^2 \\
\fl && = -\frac{72 \beta ^{12}+588 \beta ^{10}+1802 \beta ^8+2599 \beta
   ^6+1802 \beta ^4+588 \beta ^2+72}{24 \left(2 \beta
   ^2+3\right)^2 \left(2 \beta ^2+5\right) \left(3 \beta
   ^2+2\right)^2 \left(5 \beta ^2+2\right)}
\eea
Setting $\beta=1$, this leads to the predictions for the minimum of the
fBm given in the text.

\end{itemize}

\subsection{Second negative moment}
\label{app:2neg}

In Eq. (\ref{contourneg}) for $k=2$, the $u_1$ integration gives again the
residue at $u_1=0$, then the $u_2$ integration picks two poles
at $u_2=1$ and $u_2=\beta^2$. Since one must avoids the pole
at $u_2=a$ we need the condition $a>\max(1,\beta^2)$ for the
existence of the moment, in which case we find:
\bea
\fl && < y^{-2}>_{\beta,a,b,n}  = \frac{\left(a+\beta ^2+b+\beta ^2 (-n)+1\right)
   \left(a^2-\beta ^2 n \left(a+\beta ^2+b+1\right)+a b+\beta
   ^2+\beta ^4 n^2\right)}{(a-1) a \left(a-\beta ^2\right)} \nonumber
\eea
leading, for $n=0$, to the expression given in the text.  We have checked that
the same expression is obtained from the formula (\ref{momentsJ1})-(\ref{momentsJ2neg}).

\section{Appendix C: numerical values of higher moments}

For completeness, we display values of higher moments for
two of our examples.

\subsection{fBm0}

\noindent

Let us give the list of even moments $\overline{y_m^k}$, $k=6,8,..14$:
\bea
\fl \{ \frac{100691}{648270}, \frac{774289013}{6275253600},\frac{130667513591}{12726214300
   80},\frac{3027227918327}{34360778612160},\frac{13262063040175909}{171651723898374720} \}
\eea
and a longer list of numerical values for $\overline{y_m^k}$, $k=6,8,..20$:
\bea
\fl \{0.155323,0.123388,0.102676,0.0881013,0.0772615,0.0688694,0.0621715,0.0566963\} \nonumber
\eea
Taking the ratio with $1/(1+k)$, the moments of the uniform distribution, we obtain
\bea
\fl \{1.08726, 1.11049, 1.12943, 1.14532, 1.15892, 1.17078, 1.18126, \
1.19062 \}
\eea
so they decay slightly slower, meaning more weight near the edges.

The cumulants
$\overline{y_m^k}^c$, $k=4,6,..12$ are given here, together with their numerical value:
\bea
\fl \left\{-\frac{7523}{735000},\frac{4426903}{810337500},-\frac{125514889189}{19610167500000
   },\frac{128185912543691}{9885864269531250},-\frac{57847493772231002501}{14386898271448
   82812500}\right\} \nonumber \\
\fl \{-0.0102354, 0.00546304, -0.0064005, 0.0129666, -0.0402085 \}
\eea
as well as the ratio to the corresponding cumulants for the uniform distribution
\bea
\{1.22824, 1.37669, 1.53612, 1.71159, 1.90626 \}
\eea
showing again a steady growth.

\subsection{GUE-CP}

Let us give the list of even moments $\overline{y_m^k}$, $k=6,8,..14$:
\bea
\fl \left\{\frac{731327}{6586272},\frac{61661759}{772972200},\frac{31888748599}{523765962720}
   ,\frac{10558018750567}{218042523218520},\frac{2969274186629889449}{7482687918209261257
   5}
   \right\}
\eea
and a longer list of numerical values for $\overline{y_m^k}$, $k=6,8,..20$:
\bea
\fl \{0.111038,0.0797723,0.0608836,0.0484218,0.0396819,0.0332698,0.0283998,0.024598\}
 \nonumber
\eea
Taking the ratio with the moments of the semi-circle distribution $\rho(x)$, we obtain
\bea
\fl {1.06017, 1.07527, 1.08599, 1.09353, 1.09873, 1.10219, 1.10432, \
1.10543}
\eea
which grow but seem to saturate.

   The cumulants
$\overline{y_m^k}^c$, $k=4,6,..12$ are given here, together with their numerical value:
\bea
\fl && \{-\frac{541}{115248},\frac{5665234}{3459233547},-\frac{33307238400767}{261905490310
   46400},\frac{6658506099368911}{3882094130126852640},\\
\fl    && -\frac{122275968148461510151943659}{34483421254841283626472130560}\}
  \nonumber \\
\fl && {-0.00469422, 0.00163771, -0.00127173, 0.00171518, -0.00354593}
\eea
as well as the ratio to the corresponding cumulants for the semi-circle distribution
\bea
\{1.20172, 1.34162, 1.48828, 1.64698, 1.82096\}
\eea
which reveal some difference between the two distributions.

\noindent

\section{ Appendix D: Normalization of Jack polynomials}
\label{app:norm}

Recalling that $\kappa=1/\alpha$ and $k=\sum_{i=1}^{\ell(\lambda)} \lambda_i$
 one finds
that the function defined in the text in (\ref{cdef}) takes the explicit form, in the two (dual)
cases $t=1$ and $t=\alpha$:
\bea \label{c1}
\fl &&
c(\lambda,\alpha,1) = \alpha^k
\prod_{i=1}^{\ell(\lambda)} (\kappa(\ell(\lambda)-i+1))_{\lambda_i}  \prod_{1 \leq i < j \leq \ell(\lambda)}
\frac{(\kappa(j-i))_{\lambda_i-\lambda_j} }{(\kappa(j-i+1))_{\lambda_i-\lambda_j} }
\eea
and also
\bea
\fl && \label{c2}
c(\lambda,\alpha,\alpha) = \alpha^k
\prod_{i=1}^{\ell(\lambda)} (\kappa(\ell(\lambda)-i)+1)_{\lambda_i}  \prod_{1 \leq i < j \leq \ell(\lambda)}
\frac{(\kappa(j-i-1)+1)_{\lambda_i-\lambda_j} }{(\kappa(j-i)+1)_{\lambda_i-\lambda_j} }
\eea
The product of these two factors being equal to the square of the norm of the Jack polynomial
$J^{(\alpha)}_\lambda$.

 \section{  Appendix E: Averages over Jacobi measure}

 \label{app:jacobi}

Recalling the definition of the Jacobi ensemble average (\ref{av})
and taking the ratio of (\ref{MKI}) to (\ref{MKIb}), we obtain the average of the
$P^{(1/\kappa)}({\bf y})$ polynomial. We want to cancel
common factors and reorder to remove the $n$ dependence
from the bounds on the product, and make it more explicit. Using the Pochhammer symbols, one obtains
\bea
 && \left\langle P^{1/\kappa}_{\lambda}({\bf y})\right\rangle_J=\prod_{i=1}^{\ell(\lambda)} \frac{(a+1+ \kappa(n-i))_{\lambda_i} }{
  (a+b+2+\kappa (2 n - i -1))_{\lambda_i} } \\
  \times
&&  \prod_{1 \leq i < j \leq \ell(\lambda)} \frac{(\kappa(j-i+1))_{\lambda_i-\lambda_j} }{ (\kappa(j-i))_{\lambda_i-\lambda_j} }
   \prod_{1 \leq i \leq \ell(\lambda) <j \leq n} \frac{(\kappa(j-i+1))_{\lambda_i} }{ (\kappa(j-i))_{\lambda_i} }
\eea
Clearly the last term can be rewritten as
\bea
&& \left\langle P^{1/\kappa}_{\lambda}({\bf y})\right\rangle_J=\prod_{i=1}^{\ell(\lambda)} \frac{(a+1+ \kappa(n-i))_{\lambda_i} }{
  (a+b+2+\kappa (2 n - i -1))_{\lambda_i} } \\
  \times
&&  \prod_{1 \leq i < j \leq \ell(\lambda)} \frac{(\kappa(j-i+1))_{\lambda_i-\lambda_j} }{ (\kappa(j-i))_{\lambda_i-\lambda_j} }
   \prod_{1 \leq i \leq \ell(\lambda)} \frac{(\kappa(n-i+1))_{\lambda_i} }{ (\kappa(\ell(\lambda)+1-i))_{\lambda_i} }
\eea
Using the identity (\ref{c1}) this simplifies into:
\bea
\fl && \left\langle P^{1/\kappa}_{\lambda}({\bf y})\right\rangle_J= \frac{\alpha^k}{c(\lambda,\alpha,1) }
\prod_{i=1}^{\ell(\lambda)} \frac{(a+1+ \kappa(n-i))_{\lambda_i} }{
  (a+b+2+\kappa (2 n - i -1))_{\lambda_i} }  (\kappa(n-i+1))_{\lambda_i}
 \eea
 which using the relation (\ref{relJP}) between the different Jack polynomials,
 gives the formula (\ref{avJ}) in the text.

\section{  Appendix F: Remark on moment formula}
\label{app:remark}

It is interesting to note that the formula (\ref{res1})-(\ref{res2}) can be rewritten as:
\bea
&& p_{(k)}(y)  = k \alpha \lim_{p \to 0}
\sum_{\lambda, |\lambda|=k} \frac{J_\lambda^{(\alpha)}(1_p)  J^{(\alpha)}_\lambda(y)}{
< J_\lambda^{(\alpha)} , J_\lambda^{(\alpha)}>}
\eea
using that $J_\lambda^{(\alpha)}(1_p) = \prod_{s \in \lambda} (p - l'_\lambda(s)+ \alpha a'_\lambda(s) )$
see Theorem 5.4 in \cite{Stanley} and (10.25) in \cite{Macdonald}. This may be compared to
the Cauchy identity \cite{Stanley}
\bea
 \sum_\lambda q^{|\lambda|} \frac{J_\lambda^{(\alpha)}({\bf y})  J^{(\alpha)}_\lambda({\bf x})}{
< J_\lambda , J_\lambda>}  = \prod_{i,j} (1- q x_i y_j)^{-1/\alpha}
\eea
where the sum is over all partitions. Another important identity is based on the so-called binomial
formula \cite{Forbook}, and is
quoted and used in \cite{MR15}. It reads
\bea
 \sum_\lambda q^{|\lambda|}
 \prod_{s \in \lambda} (a \alpha - l'_\lambda(s)+ \alpha a'_\lambda(s) )
   \frac{J^{(\alpha)}_\lambda({\bf x})}{ < J_\lambda , J_\lambda>}  = \prod_{i=1}^n (1- q x_i)^{-a}
\eea
and agrees with Cauchy formula setting $a= n/\alpha$ and all $y_i=1$.

Finally let us note an alternative way to recover the moment formula (\ref{momentsJ1})-(\ref{momentsJ2}),
starting from Equation (2.11)-(2.12) in \cite{Olshanskiwz}
\footnote{we thank A. Borodin for pointing out this reference and useful comments
about this formula.} and performing
some manipulations. One first sets
$w=0$ which selects positive signatures, i.e. usual partitions $\lambda$. Next, one applies the relation to $n$ variables $u_i$ in $[0,1]$ rather
than on the unit circle, and one applies the identities for $n < \ell(\lambda)$ (
$\ell(\lambda)$ is called $N$ there), i.e. we set the remaining variables to zero.
Next, one transforms $u_i \to q u_i$ and use homogeneity of Macdonald polynomials.
One further expands at small $z$ and consider the $O(z)$ term, which is
then averaged over the Jacobi measure using the Kadell integral.
As we have checked explicitly, this leads, after some algebra,
to (\ref{momentsJ1})-(\ref{momentsJ2}).

\section{ Appendix G: Distribution of the value of the minimum}
\label{app:min}

It is useful to recall the analysis of \cite{FLDR09} for the PDF of
the value of the maximum, but in a much more
concise form, and further elaborate to the present cases. We will
not attempt at rigor and refer the reader to
\cite{Ostrovsky1,Ostrovsky2,Ostrovsky3,Ostrovsky4}
for steps in that direction.

\subsection{Main result}

Let us first present the quick and dirty version, directly at $\beta=1$
and give later a better justification starting from $\beta<1$ and using duality.
The positive integer moments of the reduced partition sum $z_\beta=\Gamma(1-\beta^2) Z_\beta$
of the model (\ref{modelab}) on the interval
are given by the Selberg integral (\ref{MKIb})
\bea
\overline{z_\beta^n} = Sl_n(\kappa,a,b)|_{\kappa=-\beta^2}
\eea
which is Eq. (7) in \cite{FLDR09}.
Let us use the identity
\bea \label{barnes1}
\prod_{j=1}^n \Gamma(z-j) = \frac{G(z)}{G(z-n)}
\eea
in terms of the Barnes function $G(x)$, valid for positive integer $n$ and
complex $z$. Brutally replace in the Selberg integral setting $\beta=1$.
For definiteness the argument in the last $\Gamma$-function in (\ref{MKIb}) may be slighlty shifted
before the replacement, and taken back to zero afterwards. This leads to
\bea
\fl && \overline{z_1^n} = \overline{e^{- n f_1}} = S(n)=\frac{G(1) G(2+a) G(2+b) G(4+a+b-2 n)}{G(1-n)
G(2+a-n) G(2+b-n) G(4+a+b-n)} \label{f1}
\eea
which is now continued to complex $n$, with $Re(n)<1$. Here $a,b$ denote the
value of the (possibly temperature-dependent) parameters at $\beta=1$.
Note that $a,b$ may also depend on $n$, e.g. see below the fBm example.

Freezing (see below) states that the PDF of the random variable
\bea
y_{\beta} = f_\beta - G/\beta
\eea
where $G$ is a unit Gumbel random variable, independent of $f_\beta$, is independent of $\beta$
for all $\beta \geq 1$. This implies
\cite{footnoteCao}
\bea
V_m =_\text{in  law} f_1 - G
\eea
hence the PDF of $V_m$, can be obtained by inverse Laplace transform of (\ref{f1}) convoluted
with Gumbel, i.e.
\bea
&& \overline{e^{-  n V_m }} = \Gamma(1-n) S(n) \\
&& Q(V_m) = LT^{-1}_{n \to V_m} \Gamma(1-n) S(n) \label{QQ}
\eea

Using that:
\bea
\partial_z \ln G(z) = \frac{1}{2} (1+ \ln (2 \pi)) -z + (z-1) \psi(z)
\eea
One easily obtain the cumulants of $V_m$. One must distinguish two cases
according to whether the Jacobi variables $a,b$ (at $\beta=1$) depend on $n$
or not (and similarly when performing the inverse Laplace transform).

In the case where $a,b$ are $n$-independent (LCGP with edge charges,
and GUE-CP)
one finds, for $p \geq 1$
\bea
&& \overline{V_m^p}^c = (2^p-1) \phi_p(4+a+b) - \phi_p(2+a) - \phi_p(2+b) + \gamma_p
\eea
where
\bea
&& \phi_p(z) = (p-1) \psi_{p-2}(z) + (z-1) \psi_{p-1}(z) \\
&& \gamma_p = (-1)^p (p-1)! (\zeta(p)+\zeta(p-1))  \quad , \quad p \geq 3  \\
&& \gamma_2 = \gamma_E + \frac{\pi^2}{6}  \quad , \quad \gamma_1 =  - \gamma_E - \ln 2 \pi
\eea
where we used, for $p \geq 2$, 
$\psi_p(1)=(-1)^{p+1} \zeta(p+1) p!$, and the
cumulants of the Gumbel distribution, $<G^p>^c=(p-1)! \zeta(p)$.
The results for the LCGP with edge charges are then obtained by setting
$a=\bar a$, $b=\bar b$, and for the GUE-CP $a=b=1$.
In this case the domain of parameters where (\ref{QQ}) yields a bona-fide, i.e. positive
and well defined PDF, has been discussed in \cite{FLDR09}. For the
GUE-CP one finds
\bea
\overline{V_m^2}^c = -\frac{629}{48}+2 \pi ^2 = 6.63504 \label{sc}   \\
\overline{V_m^3}^c = -64 \zeta (3)+\frac{50549}{864}+\frac{4 \pi ^2}{3} = -5.26638  \\
\overline{V_m^4}^c = -72 \zeta (3)-\frac{423301}{1152}+\frac{24 \pi ^4}{5} = 13.5668
\eea
in agreement \cite{footnoteNick} with the study \cite{FyoSim15}.

For the second case let us discuss the fBm0, where $a=2 \beta^2 n \to 2 n$ should be
inserted. This problem is much more delicate. Indeed for $a$ sufficiently
negative we know that a binding transition occurs at the boundary \cite{FLDR09}.
The naive application of the method
would yield, for the Laplace transform of the PDF of the value of the minimum $V_m$
\bea
\fl && \overline{e^{- n V_m} } = \Gamma[1-n] \frac{G(1) G(2+2 n ) G(2) G(4)}{G(1-n)
G(2+n) G(2-n) G(4+n)} \label{f1}
\eea
where $G(1)=G(2)=1$ and $G(4)=2$. This however cannot be the LT of a positive probability,
since it vanishes at $n=-1$
and convexity is violated around $n < - .22..$. Worse, we find
that the Taylor coefficient of $n^4$ is $-0.116681$, thus cumulants
cannot be obtained. The origin of this problem is clear. Contrarily to
the other cases the value of the potential is fixed at $x=0$, $V(0)=0$.
Hence $V_m$ cannot be positive, it is either $V_m<0$ or pinned at $V_m=0$ in which case
$x_m=0$. When $n$ is negative, the moment (\ref{f1}) gives a lot of weight to higher
(less negative values) of $V_m$, and eventually it reaches $V_m=0$ corresponding to the above mentioned
"binding transition". This explains why the analytical
structure of (\ref{f1}) is badly behaved on the negative $n$ side, and cannot be
the proper analytical continuation in that region. At the minimum one needs
to better take into account a possible delta-function weight at $V_m=0$.
Fixing this problem, and finding the correct analytical continuation for this case,
seems challenging and is left for future studies
\cite{footnoteDima}.

\subsection{Duality and freezing}

Let us now consider $\beta<1$ where analytical continuation
can be studied with more care. There exists a function $\tilde G_\beta(x)$ (see
\cite{noteBarnes}) such that
(\ref{barnes1}) generalizes to
\bea \label{barnes2}
\prod_{j=1}^n \Gamma(\beta z- j \beta^2) = \frac{\tilde G_\beta(z)}{\tilde G_\beta(z-n \beta)}
\eea
where $\tilde G_\beta(z) = A_{z,\beta} G_\beta(z)$ and $G_\beta(z)$ is the function
introduced in \cite{FLDR09} which satisfies the duality invariance
$G_\beta(z) = G_{1/\beta}(z)$, and $G_1(z)=G(z)$, $A_{z,1}=1$.
The precise value of $A_{z,\beta} = \beta^{\frac{z^2}{2} - \frac{z}{2} (\beta + \frac{1}{\beta})}
(2 \pi)^{z(\frac{1}{2 \beta}-\frac{1}{2})}$. Using this formula with $a=\beta \bar a$ and $b=\beta \bar b$
we obtain from the Selberg integral (\ref{MKIb})
\bea
\fl && \overline{z_\beta^n} = \overline{e^{- n \beta f_\beta}} = S_\beta(n) \label{f2} \\
\fl && = \frac{\tilde G_\beta(\frac{1}{\beta}) \tilde G_\beta(\frac{1}{\beta} + \beta + \bar a) \tilde G_\beta(\frac{1}{\beta} + \beta +\bar b)
\tilde G_\beta(\frac{2}{\beta} + 2 \beta + \bar a + \bar b -2 \beta n)}{\tilde G_\beta(\frac{1}{\beta}- \beta n)
\tilde G_\beta(\frac{1}{\beta} + \beta + \bar a - \beta n) \tilde G_\beta(\frac{1}{\beta} + \beta + \bar b - \beta n)
\tilde G_\beta(\frac{2}{\beta}  + 2 \beta + \bar a + \bar b - \beta n)} \nonumber
\eea
which is now legitimate for $\beta$ small enough at fixed $n$. Now, using
that $\tilde G_\beta(z)=\tilde G_\beta(z+\beta)/\Gamma(\beta z)$, and upon
a trivial shift $z=\tilde z (2 \pi)^{1-\beta}$ (and further ignoring the tilde)
we note that it can be rewritten as
\bea
&& \overline{z_\beta^n} = \overline{e^{- n \beta f_\beta}} =  \Gamma(1-\beta^2 n)  \bar S_\beta(\beta n)
\label{f3}
\eea
where
\bea
\fl && \bar S_\beta(\beta n) = \frac{
G_\beta(\beta+\frac{1}{\beta}) G_\beta(\frac{1}{\beta} + \beta + \bar a) G_\beta(\frac{1}{\beta} + \beta +\bar b)
G_\beta(\frac{2}{\beta} + 2 \beta + \bar a + \bar b -2 \beta n)}{G_\beta(\beta+\frac{1}{\beta}- \beta n)
G_\beta(\frac{1}{\beta} + \beta + \bar a - \beta n) G_\beta(\frac{1}{\beta} + \beta + \bar b - \beta n)
G_\beta(\frac{2}{\beta}  + 2 \beta + \bar a + \bar b - \beta n)}  \nonumber
\eea
is a fully duality invariant function $S_\beta(x) = S_{1/\beta}(x)$. Here we are using the same definition
of duality invariance as in (\ref{duality}), i.e. the combination $\beta n$ is duality invariant
(since $n'=n\beta^2=n\beta/\beta'$).
We note that the factor $\Gamma(1-\beta^2 n)$ corresponds exactly to the moments
of the (simpler) log-circular ensemble \cite{FB08}, in other words in that case
the factor $\bar S(\beta n)=1$.

So we see from (\ref{f3}), and since $n \beta$ is duality invariant, that {\it the
free energy random variable is not duality invariant.} However it is now trivial to
see what one must do to make (\ref{f3}) fully duality invariant, namely
multiply it by $\Gamma(1-n)$, which is the image of $\Gamma(1-n \beta^2)$
under duality. And this amounts precisely to a convolution by an independent Gumbel variable.
Hence defining the random variable
\bea \label{1}
y_\beta = f_\beta - G/\beta
\eea
we have that
\bea
\overline{e^{- n \beta y_\beta} } = \Gamma(1-n) \overline{e^{- n \beta f_\beta} }
= \Gamma(1-n) \Gamma(1- n \beta^2) \bar S(\beta n)
\eea
and the random variable $y_\beta$ is fully duality invariant (meaning all its exponential
moments are).
Now the freezing duality conjecture states that the PDF of the
variable $y_\beta$ freezes at $\beta=1$, leading to the results for the PDF of the
minimum $V_m$, displayed in
the previous subsection (since $G_1(z)=\tilde G_1(z)=G(z)$ all
formula trivially match). Note that the CDF of $y_\beta$ is
$1- g_\beta(y)$ where $g_\beta(y)=\overline{e^{-e^{\beta (y-f_\beta)}}}$,
hence duality invariance and freezing of (\ref{1}) is equivalent to stating that the full function
$g_\beta(y)$ is duality invariant and freezes.

Finally note that all which is needed for duality invariance
and freezing in this class of models, is that,
as discussed in (\ref{aga}), $\bar a$ and $\bar b$ are duality invariant
functions of $\beta$ and (when it happens) of $\beta n$, which is
the case for all three examples studied here.

\section{ Appendix H: Joint distribution and correlations between value and positions}
\label{app:joint}

In this Appendix we clarify the information encoded in the $n$-dependence
of the moments, for the disordered model, and show how it is related to the joint
distribution of position and value of the global extremum.

\subsection{High temperature phase}

Consider the LCGP with edge charge with $a=\beta \bar a$ and $b= \beta \bar b$ in
order to have a duality invariant problem and the simplest example, i.e. the first
moment
\bea
 <y>_{\beta,a,b,n} = \frac{\beta^{-1} + \bar a + \beta - \beta n}{2
\beta^{-1} + \bar a + \bar b + 2 \beta - 2 \beta n}
\eea
 and its Laguerre ensemble
limit, which is even simpler
 \bea
 <y>_{L,\beta,a,n} = \beta^{-1} + \bar a + \beta -
\beta n
\eea

These moments are a function of $\beta n$.
It is actually a general property, so let
denote in general $M_k(\beta n)$ the dependence of the
$k$-th moment as a function of $\beta n$. By definition of these moments
one has:
\bea
M_k(\beta n) \overline{Z_V^n} = \overline{ <y^k>_V Z_V^n}
\eea
where we have suppressed the label $\beta$ but emphasized the
dependence in the realization, noted $V$, of the random potential.
Defining the sample dependent free energy $f_V = - \beta^{-1} \ln Z$, this
can be rewritten as
\bea \label{idf}
M_k(\beta n) \overline{e^{- \beta n f_V}} = \overline{ <y^k>_V e^{- \beta n f_V}}
\eea
It is easy to see that this is equivalent to:
\bea
M_k(\partial_f) \overline{\delta(f-f_V)} = \overline{ <y^k>_V \delta(f-f_V) }
\eea
since (\ref{idf}) is recovered by multiplying by $e^{- \beta f n}$, integrating
over the full real axis, and performing multiple integrations by part (e.g. of can think
for instance as $M_k(z)$ as given by its Taylor expansion around $z=0$).

This equation has various consequences. First one can multiply
by $f^q$ and integrate over $f$. This leads to a series of equalities
which allow to interpret the
coefficients the expansion in powers of $(\beta n)^2$ as
correlations between the moment and the free energy as
\bea
 - M'_k(0) &=& \overline{<y^k> f} - \overline{<y^k>} ~~ \overline{ f} =  \overline{( <y^k> - \overline{ <y^k>} )(f - \overline{f}) }  \\
M_k''(0)&=& \overline{( <y^k> - \overline{ <y^k>} )(f - \overline{f})^2 }
\eea
and so on.

For instance for the above case one obtains the cross-correlation
\bea
\overline{<y> f} - \overline{<y>} ~~ \overline{ f}
= \frac{\beta ^2 (\bar b- \bar a)}{(\beta  (\bar a+\bar b+2 \beta )+2)^2}
\eea
which is equal to $1$ in the Laguerre ensemble model.

Second, defining the PDF of the free energy
\bea
P(f) = \overline{\delta(f-f_V)}
\eea
we obtain the conditional expectation:
\bea
&& \mathbb{E}(<y^k>|f) = \overline{<y^k>}_f = \frac{1}{P(f)}
\overline{ <y^k>_V \delta(f-f_V) }  \\
&& = \frac{1}{P(f)} M_k(\partial_f) P(f)
\eea

For instance for the Laguerre ensemble we
obtain
\bea
\mathbb{E}(<y^k>|f) = \beta^{-1} + \bar a + \beta - \frac{P'(f)}{P(f)}
\eea

\subsection{Freezing}

Since it is not the PDF of the free energy, $f$, which freezes,
we must introduce the proper generating function which freezes, which
is known to be
\bea
g_\beta(v) = \sum_{n=0}^\infty \frac{1}{n!} e^{\beta n v} \overline{Z_V^n}
= \overline{ e^{- e^{\beta v} Z_V}}
\eea
(called $g_\beta(y)$ in our previous work \cite{FLDR09}),
such that $1- g_\beta(s)$ is the CDF of the variable, called $v_\beta= f_\beta - G/\beta$, which
freezes. Similar manipulations as above lead to:
\bea
M_k(\partial_v) g_\beta(v)  = \overline{ <y^k>  e^{- e^{\beta v} Z_V} }
\eea
Since the combination $\beta n$ is duality invariant (as can be seen from
\ref{duality}) the function $M_k(z)$ freezes at $\beta=1$
(let us denote $\bar M_k(z)= M_k(z)|_{\beta=1}$ its
value), and so
does $g_\beta(v)$. Hence we can take the limit $\beta \to \infty$
on the r.h.s and we find, introducing $V_m$ the value of the
global extremum
\bea
\overline{ y_m^k  \delta(V_m - v) } = \bar M_k(\partial_v) Q(v)
\eea
in terms of the PDF of the value of the minimum
\bea
Q(v)=\overline{\delta(V_m-v) } = - \partial_v g_{\beta=1}(v)
\eea

We thus again obtain relations such that
\bea
 - \bar M'_k(0) &=& \overline{y_m^k V_m} - \overline{y_m^k} ~~ \overline{ V_m} =
 \overline{(y_m^k - \overline{y_m^k} ) (V_m-\overline{ V_m} ) }  \\
 \bar M_k''(0)
&=& \overline{ (y_m^k - \overline{ y^k_m} ) (V_m - \overline{V_m})^2 }
\eea
The more general formula being obtained from Taylor expansion in powers of $q$ of
\bea
\bar M_k(-q) = \overline{ y^k_m e^{q (V_m - \overline{V_m})} } e^{- \sum_{p \geq 1}
\frac{q^p}{p!} \overline{(V_m - \overline{V_m})^p}^c }
\eea

In particular we find the conditional expectation
\bea \label{momcond1}
 \mathbb{E}(y_m^k| V_m) 
= \frac{1}{Q(V_m)} \bar M_k(\partial_{V_m})  Q(V_m)
\eea
a formula very similar to the one obtained above for the free energy, but involving now information
about the extremum.

An interesting example is the first moment in the fBM. Then one has $\bar M_1(x)= \frac{1}{2} +
\frac{1}{4} x$. This implies
\bea
 \mathbb{E}(y_m | V_m) = \frac{1}{2} + \frac{1}{4} \frac{Q'(V_m)}{Q(V_m)}
\eea
which also implies
\bea
\overline{(y_m - \overline{y_m}) V_m^p} = - \frac{p}{4} \overline{V_m^{p-1}}
\eea
for any positive integer $p$.

For general charges $\bar a, \bar b$ the function $g_{\beta=1}(y)$ can be
obtained from our previous work \cite{FLDR09}, as is summarized in the
Appendix G above.


\begin{thebibliography}{99}
\bibitem{noteBarnes} Note that the function $\tilde G_\beta(x)$ is related to
the Alexeiewsky-Barnes G function noted $G(z|\tau)$ defined in Eq. (8) of
\cite{Ostrovsky2} (and studied there) as follows $\tilde G_\beta(\beta x)/ \tilde G_\beta(\beta) = G(x|\tau)|_{\tau=1/\beta^2}$.
\end{thebibliography}
\end{document}